\documentclass[11pt]{article}
\usepackage{graphicx,amsmath,amsfonts,amscd,amssymb,bm,cite,epsfig,epsf,url,
color,algorithm,algorithmic, mathrsfs}
\usepackage{fullpage} \usepackage[small,bf]{caption}
\usepackage{subfigure}
\newtheorem{theorem}{Theorem}[section]
\newtheorem{lemma}[theorem]{Lemma}

\newtheorem{definition}[theorem]{Definition}

\newcommand{\R}{\mathbb{R}}

\newcommand{\<}{\langle}
\renewcommand{\>}{\rangle}

\newcommand{\sgn}{\textrm{sgn}}

\newcommand{\I}{\mathcal{I}}
\renewcommand{\P}{\operatorname{\mathbb{P}}}
\newcommand{\PP}{\mathcal{P}}
 \newcommand{\PGj}{\mathcal{P}_{\Gamma_j}}
\newcommand{\PG}{\mathcal{P}_\Gamma}

\newcommand{\PT}{\mathcal{P}_T}
\newcommand{\PTp}{\mathcal{P}_{T^\perp}}

\newcommand{\POzero}{\mathcal{P}_{\Omega_0}}

\newcommand{\supp}[1]{\operatorname{supp}(#1)}

\numberwithin{equation}{section}

\def \endprf{\hfill {\vrule height6pt width6pt depth0pt}\medskip}
\newenvironment{proof}{\noindent {\bf Proof} }{\endprf\par}

\title{Compressed Sensing and Matrix Completion with Constant Proportion of Corruptions}

\author{ Xiaodong Li\\
  \vspace{-.1cm}\\
  Department of Mathematics, Stanford University, Stanford, CA 94305\\
}

\date{}

\begin{document}

\maketitle

\vspace{-0.3in}

\begin{abstract}
  In this paper we improve existing results in the field of compressed sensing and
  matrix completion when sampled data may be grossly corrupted. We
  introduce three new theorems. 1) In compressed sensing, we show that
  if the $m \times n$ sensing matrix has independent Gaussian entries,
  then one can recover a sparse signal $x$ exactly by tractable
  $\ell_1$ minimization even if a positive fraction of the
  measurements are arbitrarily corrupted, provided the number of
  nonzero entries in $x$ is $O(m/(\log (n/m)+1))$. 2) In the very
  general sensing model introduced in \cite{CP10} and assuming a
  positive fraction of corrupted measurements, exact recovery still
  holds if the signal now has $O(m/(\log^2 n))$ nonzero entries. 3)
  Finally, we prove that one can recover an $n \times n$ low-rank
  matrix from $m$ corrupted sampled entries by tractable optimization
  provided the rank is on the order of $O(m/(n\log^2 n))$; again, this
  holds when there is a positive fraction of corrupted samples.
\end{abstract}


\maketitle

\section{Introduction}

\subsection{Introduction on Compressed Sensing with Corruptions}
Compressed sensing (CS) has been well-studied in recent years
\cite{CRT06,Donoho06}. This novel theory asserts that a sparse or
approximately sparse signal $x\in \mathbb{R}^n$ can be acquired by
taking just a few non-adaptive linear measurements. This fact has
numerous consequences which are being explored in a number of
fields of applied science and engineering. In CS, the acquisition
procedure is often represented as $y=Ax$, where $A \in \mathbb{R}^{m
  \times n}$ is called the sensing matrix and $y \in \mathbb{R}^m$ is
the vector of measurements or observations. It is now well-established
that the solution $\hat{x}$ to the optimization problem
\begin{equation}
\min_{\tilde{x}} \|\tilde{x}\|_1 \text{~~such that~~} A\tilde{x}=y,
\end{equation}
is guaranteed to be the original signal $x$ with high probability,
provided $x$ is sufficiently sparse and $A$ obeys certain conditions. A
typical result is this: if $A$ has iid Gaussian entries, then exact
recovery occurs provided $\|x\|_0\leq C m/(\log(n/m)+1)$ \cite{CRT06CPAM, Donoho06CPAM, RV06} for some positive numerical constant
$C > 0$. Here is another example, if $A$ is a matrix with rows
randomly selected from the DFT matrix, the condition becomes $\|x\|_0
\leq C m/\log n$ \cite{CRT06}. \\
\\
This paper discusses a natural generalization of CS, which we shall refer to as
{\em compressed sensing with corruptions}.  We assume that some
entries of the data vector $y$ are totally corrupted but we have
absolutely no idea which entries are unreliable. We still want to
recover the original signal efficiently and accurately.  Formally, we
have the mathematical model
\begin{equation}
\label{encoding with corruptions}
y=Ax+f=[A, I]\begin{bmatrix} x\\ f \end{bmatrix},
\end{equation}
where $x \in \mathbb{R}^n$ and $f \in \mathbb{R}^m$. The number of
nonzero coefficients in $x$ is $\|x\|_0$ and
similarly for $f$. As in the above model, $A$ is an $m \times n$
sensing matrix, usually sampled from a probability distribution. The
problem of recovering $x$ (and hence $f$) from $y$ has been recently studied in the literature in connection with  some interesting applications. We discuss a few of them.
\begin{itemize}

\item {\em Clipping.} Signal clipping frequently appears because of
  nonlinearities in the acquisition device \cite{LBDB11, SKPB11}. Here, one typically measures $g(Ax)$ rather than
  $Ax$, where $g$ is always a nonlinear map. Letting $f=g(Ax)-Ax$, we
  thus observe $y =Ax + f$. Nonlinearities usually occur at large
  amplitudes so that for those components with small amplitudes, we
  have $f = g(Ax)- Ax = 0$. This means that $f$ is sparse and,
  therefore, our model is appropriate. Just as before, locating the
  portion of the data vector that has been clipped may be difficult
  because of additional noise.

\item {\em CS for networked data.} In a sensor
  network, different sensors will collect measurements of the same
  signal $x$ independently (they each measure $z_i=\langle a_i,
  x\rangle$) and send the outcome to a center hub for analysis
  \cite{HBRN08, NT11}. By setting $a_i$ as the row vectors of
  $A$, this is just $z=Ax$. However, typically some sensors will fail to
  send the measurements correctly, and will sometimes report totally
  meaningless measurements. Therefore, we collect $y=Ax+f$, where $f$
  models recording errors.
\end{itemize}

There have been several theoretical papers investigating the exact
recovery method for CS with corruptions \cite{WM10, LWW10, LDB09, NT11, SKPB11}, and all of them consider
the following recovery procedure in the noiseless case:
\begin{equation}
\label{decoding 2} 
\min_{\tilde{x}, \tilde{f}} \|\tilde{x}\|_1+\lambda(m, n)\|\tilde{f}\|_1 \text{~~such that~~}
A\tilde{x}+\tilde{f}=[A, I]\begin{bmatrix} \tilde{x}\\ \tilde{f}
\end{bmatrix}=y.
\end{equation}
We will compare them with our results in Section 1.4.\\

\subsection{Introduction on matrix completion with corruptions}

Matrix completion (MC) bears some similarity with CS. Here, the goal is to recover a low-rank matrix $L\in
\mathbb{R}^{n\times n}$ from a small fraction of linear measurements.
For simplicity, we suppose the matrix is square as above (the general
case is similar).  The standard model is that we observe $\PP_O(L)$ where $O \subset [n]\times [n]:=\{1,...,n\} \times \{1,...,n\}$ and
\[
\PP_O(L)_{ij}=
\begin{cases}
L_{ij} \text{~~if~~}(i,j)\in O;\\
0 \text{~~otherwise}.\\
\end{cases}
\]
The problem
is to recover the original matrix $L$, and there have been many papers
studying this problem in recent years, see \cite{RFP10, CR09, CT10, KMO10, Gross09}, for example. Here one
minimizes the nuclear norm --- the sum of all the singular values \cite{Fazel02}---
to recover the original low rank matrix.  We discuss
below an improved result due to Gross
\cite{Gross09} (with a slight difference).\\
\\
Define $O \sim \text{Ber}(\rho)$ for some $0<\rho<1$ by meaning that $1_{\{(i,j) \in O\}}$ are iid Bernoulli random variables with parameter $\rho$. Then the solution to
\begin{equation}
\label{L decoding} 
\min_{\widetilde{L}} \|\widetilde{L}\|_* \text{~~such that~~} \PP_O (\widetilde{L})=\PP_O(L),
\end{equation}
is guaranteed to be exactly $L$ with high probability, provided $\rho \geq {{C_\rho r \mu \log^2 n}\over n}$.
Here, $C_\rho$ is a positive numerical constant, $r$ is the rank of $L$, and $\mu$ is an incoherence parameter introduced  in \cite{CR09} which is only dependent of $L$.\\
\\
This paper is concerned with the situation in which some entries may
have been corrupted. Therefore, our model is that we observe
\begin{equation}
\label{L encoding with corruptions}
\PP_O(L)+S,
\end{equation}
where $O$ and $L$ are the same as before and $S \in \R^{n \times n}$ is supported on $\Omega \subset O$. Just as in
CS, this model has broad applicability. For example,
Wu et al.~used this model in photometric stereo \cite{WGSMWM10}. This
problem has also been introduced in \cite{CLMW} and is related to
recent work in separating a low-rank from a sparse component
\cite{CSPW09, CLMW, HKZ11, CSPW09IFAC, XCS10}. A typical result is that the solution
$(\widehat{L}, \widehat{S})$ to
\begin{equation}
\label{decoding 3} 
\min_{\widetilde{L}, \widetilde{S}} \|\widetilde{L}\|_*+\lambda(m, n)\|\widetilde{S}\|_1 \text{~~such that~~}
\PP_O(\widetilde{L})+\widetilde{S}=\PP_O(L)+S,
\end{equation}
is guaranteed to be the true pair $(L, S)$ with high probability under some assumptions about $L, O, S$ \cite{CLMW, CJSC11}. We will compare them with our result in Section 1.4.

\subsection{Main results}

This section introduces three models and three corresponding recovery
results. The proofs of these results are deferred to Section 2 for
Theorem 1.1, Section 3 for Theorem 1.2 and Section 4 for Theorem 1.3.

\subsubsection{CS with iid matrices [Model 1]}

\begin{theorem}
\label{thm 1.1} 
Suppose that $A$ is an
$m \times n$ $(m<n)$ random matrix whose entries are iid Gaussian variables
with mean $0$ and variance $1/m$, the signal to acquire is $x \in \R^n$, and our observation is $y=Ax+f+w$ where $f, w \in \R^m$ and $\|w\|_2\leq \epsilon$.
Then by choosing $\lambda(n,m)={1\over{\sqrt{\log (n/m)+1}}}$, the solution $(\hat{x}, \hat{f})$ to
\begin{equation}
\label{decoding 1}
\min_{\tilde{x}, \tilde{f}} \|\tilde{x}\|_1+\lambda \|\tilde{f}\|_1 \text{~~such that~~}
\|(A\tilde{x}+\tilde{f})-y)\|_2 \leq \epsilon
\end{equation}
satisfies  $\|\hat{x}-x\|_2+\|\hat{f}-f\|_2\leq K\epsilon$ with
probability at least $1-C\exp(-cm)$. This holds universally; that is
to say, for all vectors $x$ and $f$ obeying $\|x\|_0\leq \alpha
m/(\log (n/m)+1)$ and $\|f\|_0\leq \alpha m$. Here $\alpha$, $C$, $c$ and $K$ are numerical constants.
\end{theorem}
In the above statement, the matrix $A$ is random. Everything else is
deterministic. The reader will notice that the number of nonzero
entries is on the same order as that needed for recovery from clean
data \cite{CRT06CPAM, Donoho06, BDDW07, RV06}, while the condition of $f$ implies
that one can tolerate a constant fraction of possibly adversarial
errors. Moreover, our convex optimization is related to LASSO \cite{Tib96} and Basis Pursuit \cite{CDS98}.

\subsubsection{CS with general sensing matrices [Model 2]} 

In this model, $m<n$ and
\[
A={1\over{\sqrt{m}}}\begin{pmatrix}a_1^*\\ ...\\ a_m^*\end{pmatrix},
\]
where $a_1,..., a_m$ are $n$ iid copies of a random vector $a$ whose
distribution obeys the following two properties: 1)
$\mathbb{E}aa^*=I$; 2) $\|a\|_\infty\leq \sqrt{\mu}$. This model has
been introduced in \cite{CP10} and includes a lot of the stochastic
models used in the literature. Examples include partial DFT matrices,
matrices with iid entries, certain random convolutions \cite{Romberg09} and so on. 
\\
\\
In this model, we assume that $x$ and $f$ in \eqref{encoding with
  corruptions} have fixed support denoted by $T$ and $B$, and with
cardinality $|T|=s$ and $|B|=m_b$. In the remainder of the paper,
$x_T$ is the restriction of $x$ to indices in $T$ and $f_B$ is the
restriction of $f$ to $B$. Our main assumption here concerns the sign
sequences: the sign sequences of $x_T$ and $f_B$ are independent of
each other, and each is a sequence of symmetric iid $\pm 1$ variables. 

\begin{theorem}
 \label{thm 1.2} 
 For the model above, the solution $(\hat{x}, \hat{f})$ to
 \eqref{decoding 2}, with $\lambda(n,m)=1/\sqrt{\log
   n}$, is exact with probability at least $1-Cn^{-3}$, provided that
 $s \leq \alpha{m\over{\mu \log^2 n}}$ and $m_b\leq
 \beta{m\over{\mu}}$. Here $C$, $\alpha$ and $\beta$ are some
 numerical constants.
\end{theorem}

Above, $x$ and $f$ have fixed supports and random signs. However, by a
recent de-randomization technique first introduced in \cite{CLMW},
exact recovery with random supports and fixed signs would also
hold. We will explain this de-randomization technique in the proof of Theorem 1.3. In some specific models, such as independent rows from the DFT matrix, $\mu$ could be a numerical constant, which implies the proportion of corruptions is also a constant. An open problem is whether Theorem 1.2 still holds in the case where $x$ and $f$ have both fixed supports and signs. Another open problem is to know whether the
result would hold under more general conditions about $A$ as in
\cite{CP09} in the case where $x$ has both random support and random signs.\\
\\
We emphasize that the sparsity condition $\|x\|_0\leq C{m\over{\mu
    \log^2 n}}$ is a little stronger than the optimal result available
in the noise-free literature \cite{CRT06, CP10}),
namely,$\|x\|_0\leq C{m\over{\mu \log n}}$. The extra logarithmic
factor appears to be important in the proof which we will explain in Section 3, and a third open problem
is whether or not it is possible to remove this factor.\\
\\
Here we do not give a sensitivity analysis for the recovery procedure as in Model 1. Actually by applying a similar method introduced in \cite{CP10} to our argument in Section 3, a very good error bound could be obtained in the noisy case. However, technically there is little novelty but it will make our paper very long. Therefore we decide to only discuss the noiseless case and focus on the sampling rate and corruption ratio.

\subsubsection{MC from corrupted entries [Model 3]}
We assume $L$ is of rank $r$ and write its reduced SVD as $L=U\Sigma V^*$, where $U, V \in \mathbb{R}^{n \times r}$ and
$\Sigma \in \mathbb{R}^{r \times r}$.  Let $\mu$ be the
smallest quantity such that for all $1 \le i \le n$, 
\[
\|UU^*e_i\|_2^2\leq {{\mu r}\over n}, \quad \|VV^*e_i\|_2^2\leq {{\mu
    r}\over n}, \quad \text{and } \|UV^*\|_\infty\leq {\sqrt{\mu
    r}\over n}.
\] 
This model is the same as that originally introduced in \cite{CR09},
and later used in \cite{Gross09,Recht09,CT10,CLMW, CJSC11}.  We observe $\PP_O(L)+S$, where $O \in [n]\times [n]$ and $S$ is supported on $\Omega \subset O$. Here we assume that $O, \Omega, S$ satisfy the following model:
\paragraph{Model 3.1:}~\\
1. Fix an $n$ by $n$ matrix $K$, whose entries are either $1$ or $-1$. \\
2. Define $O \sim \text{Ber}(\rho)$ for a constant $\rho$ satisfying $0<\rho<{1 \over 2}$. Specifically speaking, $1_{\{(i,j) \in O\}}$ are iid Bernoulli random variables with parameter $\rho$. \\
3. Conditioning on $(i,j) \in O$, assume that $(i,j) \in \Omega$ are independent events with $\P((i,j) \in \Omega | (i,j) \in O)=s$. This implies that $\Omega \sim \text{Ber}(\rho s)$.\\
4. Define $\Gamma:= O/\Omega$. Then we have $\Gamma \sim \text{Ber}(\rho(1-s))$\\
5. Let $S$ be supported on $\Omega$, and $\sgn(S):=\PP_{\Omega}(K)$.\\
\begin{theorem}
\label{thm 1.3}
Under Model 3.1, suppose $\rho>C_\rho{{\mu r\log^2n}\over{n}}$ and $s \le C_s$. Moreover, suppose $\lambda:={1\over \sqrt{\rho n \log n}}$ and denote $(\hat{L}, \hat{S})$ as the optimal solution to the problem \eqref{decoding 3}. Then we have $(\hat{L}, \hat{S})=(L, S)$ with probability at least $1-Cn^{-3}$ for some numerical constant $C$, provided the numerical constants $C_s$ is sufficiently small and $C_\rho$ is sufficiently large. 
\end{theorem}

In this model $O$ is available while $\Omega$, $\Gamma$ and $S$ are not known explicitly from the observation $\PP_{O}(L)+S$. By the assumption $O \sim \text{Ber}(\rho)$, we can use $|O|/(n^2)$ to approximate $\rho$. From the following proof we can see that $\lambda$ is not required to be ${1\over \sqrt{\rho n \log n}}$ exactly for the exact recovery. The power of our result is that one can recover a low-rank matrix from a nearly minimal number of samples even when a constant proportion of these samples has been corrupted. \\
\\
We only discuss the noiseless case for this model. Actually by a method similar to \cite{CP09}, a suboptimal estimation error bound can be obtained by a slight modification of our argument. However, it is of little interest technically and beyond the optimal result when $n$ is large. There are other suboptimal results for matrix completion with noise, such as \cite{ANW11}, but the error bound is not tight when the additional noise is small. We want to focus on the noiseless case in this paper and leave the problem with noise for future work.\\
\\
The values of $\lambda$ are chosen for theoretical guarantee of exact recovery in Theorem 1.1, 1.2 and 1.3. In practice, $\lambda$ is usually taken by cross validation.

\subsection{Comparison with existing results, relative works and our
contribution}

In this section we will compare Theorems \ref{thm 1.1}, \ref{thm 1.2} and \ref{thm 1.3}
with existing results in the literature.\\
\\
We begin with Model 1. In \cite{WM10}, Wright and Ma discussed a model where the sensing matrix $A$ has independent columns with common mean $\mu$ and normal perturbations with variance $\sigma^2/m$. They chose $\lambda(m, n)=1$, and proved that
$(\hat{x}, \hat{f})=(x, f)$ with high probability provided $\|x\|_0\leq C_1(\sigma, n/m)m$, $\|f\|_0\leq C_2(\sigma, n/m)m$ and $f$ has random signs. Here $C_1(\sigma, 1/m)$ is much smaller than $C/(\log(n/m)+1)$. We notice that since the authors of  \cite{WM10} talked about a different model, which is motivated by \cite{WYGSM09}, it may not be comparable with ours directly. However, for our motivation of CS with corruptions, we assume $A$ satisfy a symmetric distribution and get better sampling rate.\\
\\
A bit later, Laska et al.~\cite{LDB09} and Li et al.~\cite{LWW10} also
studied this problem. By setting $\lambda(m, n)=1$, both papers
establish that for Gaussian (or sub-Gaussian) sensing matrices $A$, if $m >
C(\|x\|_0+\|f\|_0)\log ((n+m)/(\|x\|_0+\|f\|_0))$, then the recovery
is exact. This follows from the fact that $[A,I]$ obeys a restricted
isometry property known to guarantee exact recovery of sparse vectors
via $\ell_1$ minimization. Furthermore, the sparsity
requirement about $x$ is the same as that found in the standard CS
literature, namely, $\|x\|_0\leq C m/(\log(n/m)+1)$. However, the
result does not allow a positive fraction of
corruptions. For example, if $m=\sqrt{n}$, we have $\|f\|_0/m \leq 2/\log{n}$, which will go to zero as $n$ goes to zero.\\
\\
As for Model 2, an interesting piece of
work \cite{NT11} (and later \cite{NNT11} on the noisy case) appeared during the preparation of this paper. These
papers discuss models in which $A$ is formed by selecting rows from an
orthogonal matrix with low incoherence parameter $\mu$, which is the minimum value such that $n|A_{ij}|^2\leq \mu$ for any $i, j$. The main result states that selecting $\lambda=\sqrt{n/(C \mu m \log n)}$ gives exact
recovery under the following assumptions: 1) the rows of $A$ are
chosen from an orthogonal matrix uniformly at random; 2) $x$ is a
random signal with independent signs and equally likely to be either
$\pm 1$; 3) the support of $f$ is chosen uniformly at random. (By the
de-randomization technique introduced in \cite{CLMW} and used in
\cite{NT11}, it would have been sufficient to assume that the signs of
$f$ are independent and take on the values $\pm 1$ with equal
probability). Finally, the sparsity conditions require $m \geq C \mu^2 \|x\|_0(\log n)^2$ and $\|f\|_0\leq C m$, which are nearly optimal, for the best known sparsity condition when $f = 0$ is $m\geq
C \mu \|x\|_0 \log n$. In other words, the result is optimal up to an
extra factor of $\mu \log n$; the sparsity condition about $f$ is of
course nearly optimal. \\
\\
However, the model for $A$ does not
include some models frequently discussed in the literature such as subsampled tight or continuous frames. Against this background, a recent paper of Cand\`es and Plan \cite{CP10}
considers a very general framework, which includes a lot of common
models in the literature. Theorem 1.2 in our paper is similar to Theorem 1 in \cite{NT11}. It assumes similar sparsity conditions, but is based on this much broader and more applicable model introduced in \cite{CP10}. Notice that, we require $m\geq C \mu \|x\|_0(\log n)^2$
whereas \cite{NT11} requires $m\geq C \mu^2 \|x\|_0(\log
n)^2$. Therefore, we improve the condition by a factor of $\mu$, which
is always at least $1$ and can be as large as $n$. However, our result
imposes $\|f\|_0\leq C m/\mu$, which is worse than $\|f\|_0\leq \gamma
m$ by the same factor. In \cite{NT11}, the parameter $\lambda$ depends
upon $\mu$, while our $\lambda$ is only a function of $m$ and
$n$. This is why the results differ, and we prefer to use a value of
$\lambda$ that does not depend on $\mu$ because in some applications,
an accurate estimate of $\mu$ may be difficult to obtain. In addition, we use different techniques of proof which the clever golfing scheme of \cite{Gross09} is exploited. \\
\\
Sparse approximation is another problem of underdetermined linear system where the dictionary matrix $A$ is always assumed to be deterministic. Readers interested in this problem (which always requires
stronger sparsity conditions) may also want to study the recent paper
\cite{SKPB11} by Studer et al. There, the authors introduce a more
general problem of the form $y=Ax+Bf$, and analyzed the performance of
$\ell_1$-recovery techniques by using ideas which have been
popularized under the name of generalized uncertainty principles in
the basis pursuit and sparse approximation literature.\\
\\
As for Model 3, Theorem \ref{thm 1.3} is a significant extension of the results
presented in \cite{CLMW}, in which the authors have a stringent requirement $\rho=0.1$. In a very recent and independent work \cite{CJSC11}, the authors consider a model where both $O$ and $\Omega$ are unions of stochastic and deterministic subsets, while we only assume the stochastic model. We recommend interested readers to read the paper for the details. However, only considering their results on stochastic $O$ and $\Omega$, a direct comparison shows that the number of samples we need is less
than that in this reference. The difference is several logarithmic factors. Actually, the requirement of $\rho$ in our paper is optimal even for clean data in the literature of MC. Finally, we want to emphasize that the random support assumption is essential in Theorem 1.3 when the rank is large. Examples can be found in \cite{HKZ11}.\\
\\
We wish to close our introduction with a few words concerning the
techniques of proof we shall use. The proof of
Theorem \ref{thm 1.1} is based on the concept of restricted
isometry, which is a standard technique in the literature of CS. However, our argument involves a generalization of the restricted isometry concept.  The proofs of Theorems \ref{thm 1.2} and
\ref{thm 1.3} are based on the golfing scheme, an elegant technique
pioneered by David Gross \cite{Gross09}, and later used in
\cite{Recht09, CLMW, CP10} to construct dual
certificates. Our proof leverages results from \cite{CLMW}. However, we contribute novel elements by finding an
appropriate way to phrase sufficient optimality conditions, which are
amenable to the golfing scheme. Details are presented in the following sections.\\

\section{A Proof of Theorem \ref{thm 1.1}}
In the proof of Theorem \ref{thm 1.1}, we will see the notation $P_T x$. Here $x$ is a $k$-dimensional vector, $T$ is a subset of $\{1, ..., k\}$ and we also use $T$ to represent the subspace of all $k$-dimensional vectors supported on $T$.
Then $P_T x$ is the projection of $x$ onto the subspace $T$, which is to
keep the value of $x$ on the support $T$ and to change other elements into
zeros. In this section we use the notation ``$\lfloor . \rfloor$" of ``floor function" to represent the integer part of any real number.\\
\\
First we generalize the concept of the restricted isometry property (RIP) \cite{CT05} for the convenience to prove our theorem:
\begin{definition}
\label{definition 2.1}
For any matrix $\Phi \in \mathbb{R}^{l \times (n+m)}$, define the RIP-constant
$\delta_{s_1, s_2}$  by the infimum value of $\delta$ such that
\[
(1-\delta)(\|x\|_2^2+\|f\|_2^2)\leq \left\|\Phi \begin{bmatrix} x\\ f \end{bmatrix}\right\|_2^2\leq (1+\delta)(\|x\|_2^2+\|f\|_2^2)
\] 
holds for any $x\in
\mathbb{R}^n$ with $|\supp{x}|\leq s_1$ and $f\in \mathbb{R}^m$ with
$|\supp{f}|\leq s_2$.
\end{definition}

\begin{lemma}
\label{lemma 2.2}
For any $x_1, x_2 \in \mathbb{R}^n$ and $f_1, f_2 \in \mathbb{R}^m$ such that
$\supp{x_1}\cap \supp{x_2}=\phi$, $|\supp{x_1}|+|\supp{x_2}|\leq s_1$ and
$\supp{f_1}\cap \supp{f_2}=\phi$, $|\supp{f_1}|+|\supp{f_2}|\leq s_2$, we
have 
\[
\left|\left\langle \Phi \begin{bmatrix} x_1\\ f_1 \end{bmatrix}, \Phi \begin{bmatrix}
x_2\\ f_2 \end{bmatrix}\right\rangle \right|\leq \delta_{s_1,
s_2}\sqrt{\|x_1\|_2^2+\|f_1\|_2^2}\sqrt{\|x_2\|_2^2+\|f_2\|_2^2}
\]
\end{lemma}
\begin{proof}
First, we suppose $\|x_1\|_2^2+\|f_1\|_2^2=\|x_2\|_2^2+\|f_2\|_2^2=1$. By the
definition of $\delta_{s_1, s_2}$, we have
\[
2(1-\delta_{s_1, s_2})\leq \left\langle \Phi \begin{bmatrix} x_1+x_2\\ f_1+f_2
\end{bmatrix}, \Phi \begin{bmatrix} x_1+x_2\\ f_1+f_2 \end{bmatrix}\right\rangle
\leq 2(1+\delta_{s_1, s_2}),
\]
and
\[
2(1-\delta_{s_1, s_2})\leq \left\langle \Phi \begin{bmatrix} x_1-x_2\\ f_1-f_2
\end{bmatrix}, \Phi \begin{bmatrix} x_1-x_2\\ f_1-f_2 \end{bmatrix}\right\rangle
\leq 2(1+\delta_{s_1, s_2}).
\]
By the above inequalities, we have
$\left|\left\langle \Phi \begin{bmatrix} x_1\\ f_1 \end{bmatrix}, \Phi \begin{bmatrix}
x_2\\ f_2 \end{bmatrix}\right\rangle\right|\leq \delta_{s_1, s_2}$, and hence by
homogeneity, we have
$\left|\left\langle \Phi \begin{bmatrix} x_1\\ f_1 \end{bmatrix}, \Phi \begin{bmatrix}
x_2\\ f_2 \end{bmatrix}\right\rangle \right|\leq \delta_{s_1,
s_2}\sqrt{\|x_1\|_2^2+\|f_1\|_2^2}\sqrt{\|x_2\|_2^2+\|f_2\|_2^2}$ without the
norm assumption.
\end{proof}

\begin{lemma}
\label{lemma 2.3}
Suppose $\Phi \in \mathbb{R}^{l \times (n+m)}$ with RIP-constant $\delta_{2s_1,
2s_2}<{1\over 18}$ ($s_1, s_2 >0$)and $\lambda$ is between ${1\over 2}\sqrt{{s_1}\over{s_2}}$
and $2\sqrt{{s_1}\over{s_2}}$. Then for any $x \in \mathbb{R}^n$ with
$|\supp{x}|\leq s_1$, any $f \in \mathbb{R}^m$ with $|\supp{f}|\leq s_2$, and any $w \in \mathbb{R}^m$ with $\|w\|_2 \leq \epsilon$ the
solution $(\hat{x}, \hat{f})$ to the optimization problem \eqref{decoding 1}
satisfies $\|\hat{x}-x\|_2+\|\hat{f}-f\|_2\leq {{4\sqrt{13+13\delta_{2s_1, 2s_2}}}\over{1-9\delta_{2s_1, 2s_2}}}\epsilon$.
\end{lemma}

\begin{proof}
Suppose $\Delta x=\hat{x}-x$ and $\Delta f = \hat{f}-f$. Then by \eqref{decoding 1} we have 
\[
\left\|\Phi \begin{bmatrix} \Delta x\\ \Delta f\end{bmatrix}\right\|_2 \leq \|w\|_2+\left\|\Phi \begin{bmatrix} \hat{x}\\ \hat{f} \end{bmatrix}-\left(\Phi \begin{bmatrix} x\\ f \end{bmatrix}+w\right)\right\|_2\leq 2\epsilon.
\] 
It is easy to check that the original $(x, f)$ satisfies the inequality constraint in \eqref{decoding 1}, so we have
\begin{eqnarray}
\label{ineq 2.1}
\|x+\Delta x\|_1+\lambda \|f+\Delta f\|_1\leq \|x\|_1+\lambda \|f\|_1.
\end{eqnarray}
Then it suffices to show $\|\Delta x\|_2+\|\Delta f\|_2\leq {{4\sqrt{13+13\delta_{2s_1, 2s_2}}}\over{1-9\delta_{2s_1, 2s_2}}}\epsilon$.

Suppose $T_0$ with $|T_0|=s_1$ such that $\supp{x}\in T_0$. Denote $T_0^c=T_1
\cup \cdots \cup T_l$ where $|T_1|=...=|T_{l-1}|=s_1$ and $|T_l|\leq s_1$.
Moreover, suppose $T_1$ contains the indices of the $s_1$ largest (in the sense
of absolute value) coefficients of $P_{T_0^c}\Delta x$, $T_2$ contains the indices of
the $s_1$ largest coefficients of $P_{(T_0\cup T_1)^c}\Delta x$, and so on.
Similarly, define $V_0$ such that $\supp{f}\subset V_0$ and $|V_0|=s_2$, and divide
$V_0^c=V_1\cup ... \cup V_k$ in the same way. By this setup, we easily have
\begin{equation}
\label{ineq 2.2}
\sum_{j\geq 2}\|P_{T_{j}}\Delta x\|_2\leq s_1^{-{1\over 2}}\|P_{T_{0}^c}\Delta
x\|_1,
\end{equation}
and
\begin{equation}
\label{ineq 2.3}
\sum_{j\geq 2}\|P_{V_{j}}\Delta f\|_2\leq s_2^{-{1\over 2}}\|P_{V_{0}^c}\Delta
f\|_1. 
\end{equation}
On the other hand, by the assumption $\supp{x}\subset T_0$ and $\supp{f}\subset V_0$, we
have, 
\begin{equation}
\label{ineq 2.4}
\|x+\Delta x\|_1=\|P_{T_0}x+P_{T_0}\Delta x\|_1+\|P_{T_0^c}\Delta
x\|_1\geq\|x\|_1-\|P_{T_0}\Delta x\|_1+\|P_{T_0^c}\Delta x\|_1,
\end{equation}
and similarly,
\begin{equation}
\label{ineq 2.5}
\|f+\Delta f\|_1\geq\|f\|_1-\|P_{V_0}\Delta f\|_1+\|P_{V_0^c}\Delta f\|_1.
\end{equation}
By inequalities (\ref{ineq 2.1}), (\ref{ineq 2.4}) and (\ref{ineq 2.5}), we have 
\begin{equation}
\label{ineq 2.6}
\|P_{T_0^c}\Delta x\|_1+\lambda\|P_{V_0^c}\Delta f\|_1\leq \|P_{T_0}\Delta
x\|_1+\lambda\|P_{V_0}\Delta f\|_1.
\end{equation}
By the definition of $\delta_{2s_1, 2s_2}$, the fact $\left\|\Phi \begin{bmatrix} \Delta x\\ \Delta f\end{bmatrix}\right\|_2\leq 2\epsilon$ and Lemma \ref{lemma 2.2}, we have
\begin{align*}
&(1-\delta_{2s_1, 2s_2})\left(\left\|P_{T_0}\Delta x+P_{T_1}\Delta x\right\|^2_2+\left\|P_{V_0}\Delta f+P_{V_1}\Delta f\right\|_2^2\right)\\
&\leq \left\|\Phi \begin{bmatrix} P_{T_0}\Delta x+P_{T_1}\Delta x \\ P_{V_0}\Delta f+P_{V_1}\Delta f \end{bmatrix}\right\|_2^2\\
&=\left\langle \Phi \begin{bmatrix} P_{T_0}\Delta x+P_{T_1}\Delta x \\ P_{V_0}\Delta f+P_{V_1}\Delta f \end{bmatrix}, \Phi \begin{bmatrix} \Delta x\\ \Delta f\end{bmatrix}-\Phi \begin{bmatrix} P_{T_2}\Delta x+...+P_{T_l}\Delta x \\ P_{V_2}\Delta f+...+P_{V_k}\Delta f \end{bmatrix} \right\rangle\\
&\leq-\left\langle \Phi \begin{bmatrix} P_{T_0}\Delta x+P_{T_1}\Delta x \\ P_{V_0}\Delta f+P_{V_1}\Delta f \end{bmatrix}, \Phi \begin{bmatrix} P_{T_2}\Delta x+...+P_{T_l}\Delta x \\ P_{V_2}\Delta f+...+P_{V_k}\Delta f \end{bmatrix} \right\rangle+2\epsilon\left\|\Phi \begin{bmatrix} P_{T_0}\Delta x+P_{T_1}\Delta x \\ P_{V_0}\Delta f+P_{V_1}\Delta f \end{bmatrix}\right\|_2\\
&\leq \delta_{2s_1, 2s_2}\left(\left\|\begin{bmatrix} P_{T_0}\Delta x \\ P_{V_0}\Delta f \end{bmatrix}\right\|_2+\left\|\begin{bmatrix} P_{T_1}\Delta x \\ P_{V_1}\Delta f \end{bmatrix}\right\|_2\right)\left(\sum_{j\geq 2}\|P_{T_j}\Delta x\|_2+\sum_{j \geq 2}\|P_{V_j}\Delta f\|_2\right)\\
&+2\epsilon\sqrt{1+\delta_{2s_1,2s_2}}\sqrt{\|P_{T_0}\Delta x\|_2^2+\|P_{T_1}\Delta x\|_2^2+\|P_{V_0}\Delta f\|_2^2+\|P_{V_1}\Delta f\|_2^2}.
\end{align*}
Moreover, since
\begin{align*}
&\sum_{j\geq 2}\|P_{T_j}\Delta x\|_2+\sum_{j \geq 2}\|P_{V_j}\Delta f\|_2\\
&\leq s_1^{-{1\over 2}}\|P_{T_{0}^c}\Delta x\|_1+s_2^{-{1\over 2}}\|P_{V_{0}^c}\Delta f\|_1&~~&\text{By (\ref{ineq 2.2}) and (\ref{ineq 2.3})}\\
&\leq 2s_1^{-{1\over 2}}(\|P_{T_{0}^c}\Delta x\|_1+\lambda\|P_{V_{0}^c}\Delta f\|_1)&~~&\text{By $\lambda>{1\over 2}\sqrt{{s_1}\over{s_2}}$}\\
&\leq 2s_1^{-{1\over 2}}(\|P_{T_{0}}\Delta x\|_1+\lambda\|P_{V_{0}}\Delta f\|_1)&~~&\text{By (\ref{ineq 2.6})}\\
&\leq 2s_1^{-{1\over 2}}(s_1^{1\over 2}\|P_{T_{0}}\Delta x\|_2+\lambda s_2^{1\over 2}\|P_{V_{0}}\Delta f\|_2)&~~&\text{By Cauchy-Schwartz inequality}\\
&\leq 4\|P_{T_0}\Delta x\|_2+4\|P_{V_0}\Delta f\|_2,&~~&\text{By $\lambda<2\sqrt{{s_1}\over{s_2}}$}
\end{align*}
we have
\begin{align*}
&\left(\left\|\begin{bmatrix} P_{T_0}\Delta x \\ P_{V_0}\Delta f \end{bmatrix}\right\|_2+\left\|\begin{bmatrix} P_{T_1}\Delta x \\ P_{V_1}\Delta f \end{bmatrix}\right\|_2\right)\left(\sum_{j\geq 2}\|P_{T_j}\Delta x\|_2+\sum_{j \geq 2}\|P_{V_j}\Delta f\|_2\right)\\
&\leq 8(\|P_{T_0}\Delta x\|_2^2+\|P_{T_1}\Delta x\|_2^2+\|P_{V_0}\Delta f\|_2^2+\|P_{V_1}\Delta f\|_2^2).
\end{align*}
Therefore, by $\delta_{2s_1, 2s_2}< 1/9$, we have
\[
\sqrt{\|P_{T_0}\Delta x\|_2^2+\|P_{T_1}\Delta x\|_2^2+\|P_{V_0}\Delta f\|_2^2+\|P_{V_1}\Delta f\|_2^2}\leq{{2\epsilon\sqrt{1+\delta_{2s_1, 2s_2}}}\over{1-9\delta_{2s_1, 2s_2}}}.
\]
Since 
\[
\sum_{j\geq 2}\|P_{T_j}\Delta x\|_2+\sum_{j \geq 2}\|P_{V_j}\Delta f\|_2\leq4\|P_{T_0}\Delta x\|_2+4\|P_{V_0}\Delta f\|_2,
\]
we have
\begin{align*}
\|\Delta x\|_2+\|\Delta f\|_2&\leq 5(\|P_{T_0}\Delta x\|_2+\|P_{V_0}\Delta f\|_2)+(\|P_{T_1}\Delta x\|_2+\|P_{V_1}\Delta f\|_2)\\
&\leq \sqrt{52}\sqrt{\|P_{T_0}\Delta x\|_2^2+\|P_{T_1}\Delta x\|_2^2+\|P_{V_0}\Delta f\|_2^2+\|P_{V_1}\Delta f\|_2^2}\\
&\leq{{4\sqrt{13+13\delta_{2s_1, 2s_2}}}\over{1-9\delta_{2s_1, 2s_2}}}\epsilon.
\end{align*}
\end{proof}

We now cite a well-known result in the literature of CS, e.g. Theorem 5.2 of
\cite{BDDW07}.

\begin{lemma}
\label{lemma 2.4}
Suppose $A$ is a random matrix defined in model 1. Then for any $0< \delta <1$,
there exist $c_1(\delta), c_2(\delta)>0$ such that with probability at least $1-2\exp(-c_2(\delta)m)$,
\[
(1-\delta)\|x\|_2^2\leq\|Ax\|_2^2\leq(1+\delta)\|x\|_2^2
\]
holds universally for any $x$ with $|\supp{x}|\leq c_1(\delta){m\over{\log {n\over m}+1}}$. 
\end{lemma}

Also, we cite a well-know result which can give a bound for the biggest singular
value of random matrix, e.g. \cite{DS01} and \cite{Vershynin10}.
\begin{lemma}
\label{lemma 2.5}
Let $B$ be an $m \times n$ matrix whose entries are independent standard normal
random variables. Then for every $t\geq 0$, with probability at least $1-2\exp(-t^2/2)$, one has $\|B\|_{2,2}\leq \sqrt{m}+\sqrt{n}+t$.
\end{lemma}

We now prove Theorem \ref{thm 1.1}:\\
\begin{proof}
Suppose $\alpha$, $\delta$ are two constants independent of $m$ and $n$, and their values will be specified later. Set $s_1=\left\lfloor \alpha{m\over{\log {n\over m}+1}} \right\rfloor$ and $s_2=\lfloor \alpha m \rfloor$. We
want to bound the RIP-constant $\delta_{2s_1, 2s_2}$ for the $(n+m)\times m$
matrix $\Phi=[A, I]$ when $\alpha$ is sufficiently small. For any $T$ with $|T|=2s_1$ and $V$ with $|V|=2s_2$, and
any $x$ with $\supp{x}\subset T$, any $f$ with $\supp{f}\subset V$, we have
\begin{eqnarray*}
\left\|[A, I]\begin{bmatrix} x\\ f \end{bmatrix}\right\|_2^2=\|Ax+f\|_2^2=\|Ax\|_2^2+\|f\|_2^2+2\langle P_V A P_T x,f\rangle.
\end{eqnarray*}
By Lemma \ref{lemma 2.4}, assuming $\alpha\leq c_1(\delta)$, with probability at least $1-2\exp(-c_2(\delta)m))$ we have 
\begin{equation}
\label{ineq 2.6.1}
(1-\delta)\|x\|_2^2\leq\|Ax\|_2^2\leq (1+\delta)\|x\|_2^2
\end{equation}
holds universally for any such $T$ and $x$.\\
\\
Now we we fix $T$ and $V$, and we want to bound $\|P_V A P_T\|_{2, 2}$. By
Lemma \ref{lemma 2.5}, we actually have
\begin{equation}
\label{ineq 2.6.2}
\|P_V A P_T\|_{2, 2}\leq{1\over{\sqrt{m}}}(\sqrt{2s_1}+\sqrt{2s_2}+\sqrt{\delta^2
m})\leq(2\sqrt{2\alpha}+\delta)
\end{equation}
with probability at least $1-2\exp(-{\delta^2 m/2})$.
Then with probability at least $1-2\exp(-{{\delta^2 m}\over 2}){n \choose 2s_1}{m \choose 2s_2}$, inequality \ref{ineq 2.6.2} holds universally for any $V$ satisfying $|V|=2s_1$ and $T$ satisfying $|V|=2s_2$. By $2s_1\leq
2\alpha{m\over{\log {n\over m}+1}}$, we have $2s_1\log ({{en}\over{2s_1}})\leq
\alpha_1 m$, where $\alpha_1$ only depends on $\alpha$ and $\alpha_1 \rightarrow 0$ as $\alpha \rightarrow 0$, and hence
${n \choose 2s_1}\leq ({{en}\over{2s_1}})^{2s_1}\leq \exp(\alpha_1 m)$.
Similarly, because $2s_2\leq 2\alpha m$, we have $2s_2\log
({{em}\over{2s_2}})\leq \alpha_2 m$, where $\alpha_2$ only depends on $\alpha$ and $\alpha_2 \rightarrow 0$ as $\alpha \rightarrow 0$, and hence ${m \choose 2s_2}\leq ({{em}\over{2s_2}})^{2s_2}\leq
\exp(\alpha_2 m)$. Therefore, inequality \ref{ineq 2.6.2} holds universally for any such $T$
and $V$ with probability at least $1-2\exp((\delta^2/2-\alpha_1-\alpha_2)m)$.\\
\\
Combined with \ref{ineq 2.6.1}, we have
\[
(1-\delta)\|x\|_2^2+\|f\|_2^2-(2\sqrt{2\alpha}+\delta)\|x\|_2\|f\|_2  
\leq \left\|[A,I]\begin{bmatrix} x \\ f \end{bmatrix}\right\|_2^2 
\leq (1+\delta)\|x\|_2^2+\|f\|_2^2+(2\sqrt{2\alpha}+\delta)\|x\|_2\|f\|_2
\]
holds universally for any such $T$, $U$, $x$ and $f$ which probability at least
$1-2\exp(-c_2(\delta)m))-2\exp((\delta^2/2-\alpha_1-\alpha_2)m)$. By choosing an appropriate $\delta$ and letting $\alpha$ sufficiently small, we have $\delta_{2s_1, 2s_2}<1/9$ with
probability at least $1-Ce^{-cm}$.\\
\\ 
Moreover, under the assumption that $\alpha \left( {m\over{\log (n/m)+1}}\right)\geq 1$, we have $s_1=\left\lfloor \alpha \left( {m\over{\log (n/m)+1}}\right) \right\rfloor>0$, $s_2=\lfloor \alpha m \rfloor>0$ and ${{1\over 2}\sqrt{{s_1}\over{s_2}}}<{1\over{\sqrt{\log {n\over m}+1}}}<{2\sqrt{{s_1}\over{s_2}}}$. Then Theorem \ref{thm 1.1} as a direct
corollary of Lemma \ref{lemma 2.3}
\end{proof}

\section{A Proof of Theorem \ref{thm 1.2} }
In this section we will encounter several absolute constants. Instead of
denoting them by $C_1$, $C_2$, ..., we just use $C$, i.e., the values of $C$ change
from line to line. Also, we will use the phrase ``with high probability" to mean
with probability at least $1-Cn^{-c}$, where $C>0$ is a numerical constant and $c=3, 4, \text{~or~}5$ depending on the context.\\
\\
Here we will use a lot of notations to represent sub-matrices and sub-vectors.
Suppose $A \in \mathbb{R}^{m \times n}$, $P \subset [m]:=\{1,...,m\}$, $Q \subset [n]$ and $i \in [n]$. We
denote by $A_{P,:}$ the sub-matrix of $A$ with row indices contained in $P$, by
$A_{:,Q}$ the sub-matrix of $A$ with column indices contained in $Q$, and by
$A_{P,Q}$ the sub-matrix of $A$ with row indices contained in $P$ and column
indices contained in $Q$. Moreover, we denote by $A_{P,i}$ the sub-matrix of $A$
with row indices contained in $P$ and column $i$, which is actually a column
vector.\\
\\
The term ``vector" means column vector in this section, and all row vectors are
denoted by an adjoint of a vector, such as $a^*$ for a vector $a$. Suppose $a$ is
a vector and $T$ a subset of indices. Then we denote by $a_T$ the restriction of
$a$ on $T$, i.e., a vector with all elements of $a$ with indices in $T$. For any vector $v$, we use $v_{\{i\}}$ to denote the $i$-th element of $v$.

\subsection{Supporting lemmas}
To prove Theorem \ref{thm 1.2} we need some supporting lemmas. Because our model
of sensing matrix $A$ is the same as in \cite{CP10}, we will cite some lemmas from it
directly.
\begin{lemma}
\label{lemma 3.1}(Lemma 2.1 of \cite{CP10})
Suppose $A$ is as defined in model 2. Let $T \subset [n]$ be a fixed set of cardinality s.
Then for $\delta >0$, $\mathbb{P}(\|A_{:,T}^*A_{:,T}-I\|_{2, 2}\geq \delta)\leq
2s\exp\left(-{m\over {\mu s}}\cdot{{\delta^2}\over {2(1+{\delta/3})}}\right)$. In
particular, $\|A_{:,T}^*A_{:,T}-I\|_{2, 2}\leq {1\over 2}$ with high probability
provided $s\leq \gamma {m\over{\mu \log n}}$, and $\|A_{:,T}^*A_{:,T}-I\|_{2, 2}\leq
{1\over {2\sqrt{\log n}}}$ with high probability provided $s\leq \gamma
{m\over{\mu \log^2 n}}$, where $\gamma$ is some absolute constant.
\end{lemma}

This Lemma was proved in \cite{CP10} by matrix Bernstein's inequality, which is first introduced by \cite{AW02}. A deep generalization is given in \cite{Tropp11}.

\begin{lemma}
\label{lemma 3.2}(Lemma 2.4 of \cite{CP10})
Suppose $A$ is as defined in model 2. Fix $T \subset [n]$ with $|T|=s$ and $v
\in \mathbb{R}^s$. Then $\|A_{:,T^c}^*A_{:,T} v\|_\infty\leq
{1\over{20\sqrt{s}}}\|v\|_2$ with high probability provided $s\leq \gamma
{m\over{\mu \log n}}$, where $\gamma$ is some absolute constant.
\end{lemma}

\begin{lemma}
\label{lemma 3.3}(Lemma 2.5 of \cite{CP10})
Suppose $A$ is as defined in model 2. Fix $T \subset [n]$ with $|T|=s$. Then
$\max_{i\in T^c}\|A_{:,T}^*A_{:,i}\|_2\leq 1$ with high probability provided
$s\leq \gamma {m\over{\mu \log n}}$, where $\gamma$ is some absolute constant.
\end{lemma}

\subsection{A proof of Theorem 1.2}
In this part we will give a complete proof of Theorem \ref{thm 1.2} with a
powerful technique called "golfing-scheme" introduced by David Gross in \cite{Gross09}, and
later in \cite{CLMW} and \cite{CP10}. Under the assumption of model 2, we additionally assume
$s\leq \alpha {m\over{\mu \log^2 n}}$ and $m_b \leq \beta {m\over \mu}$, where
$\alpha$ and $\beta$ are numerical constants whose values will specified later.\\
\\
First we give two useful inequalities. By replacing $A$ with
$\sqrt{m\over{m-m_b}}A_{B^c,T}$ in Lemma 3.1 and Lemma 3.2, we have 
\begin{equation}
\label{ineq 3.2.1}
\left\|{m\over{m-m_b}}A_{B^c, T}^*A_{B^c, T}-I\right\|_{2, 2}\leq {1/2}
\end{equation} 
and 
\begin{equation}
\label{ineq 3.2.2}
\max_{i \in T^c}\left\|{m\over{m-m_b}}A_{B^c, T}^*A_{B^c,i}\right\|_2\leq 1
\end{equation}
with high probability provided $s\leq \gamma {{m-m_b}\over{\mu \log n}}$. Since
$s\leq \alpha {m\over{\mu \log^2 n}}$ and $m_b\leq \beta {m\over \mu}$, both
\ref{ineq 3.2.1} and \ref{ineq 3.2.2} hold with high probability provided
$\alpha$ and $\beta$ are sufficiently small. We assume (\ref{ineq 3.2.1}) and
(\ref{ineq 3.2.2}) hold throughout this section.\\
\\
First we prove that the solution $(\hat{x}, \hat{f})$ of \eqref{decoding 2} equals $(x, f)$ if we can find an appropriate dual vector
$q_{B^c}$ satisfying the following requirement. This is actually an ``inexact dual vector" of the optimization problem \eqref{decoding 2}. This idea was first given explicitly in \cite{GLFBE10} and \cite{Gross09}, and related to \cite{CP09Proceedings}. We give a result similar to \cite{CP10}.

\begin{lemma}{(Inexact Duality)}
\label{lemma 3.4}
Suppose there exists a vector $q_{B^c}\in \mathbb{R}^{m-m_b}$ satisfying
\begin{equation}
\label{ineq 3.2.3}
\|v_T-\sgn(x_T)\|_2\leq \lambda/4,  \text{~~~} \|v_{T^c}\|_\infty\leq 1/4
\text{~~~and~~~} \|q_{B^c}\|_\infty\leq \lambda/4,
\end{equation}
where
\begin{equation}
 \label{eq 3.2.4}
v=A_{B^c,:}^*q_{B^c}+A_{B,:}^*\lambda \sgn(f_B).
\end{equation}
Then the solution $(\hat{x}, \hat{f})$ of \eqref{decoding 2}
equals $(x, f)$ provided $\beta$ is sufficiently small and $\lambda<{3\over 2}$.
\end{lemma}

\begin{proof}
Set $h=\hat{x}-x$. By $x_{T^c}=0$ we have 
\begin{equation}
\label{eq 3.2.5}
h_{T^c}=\hat{x}_{T^c}.
\end{equation}
 By $f_{B^c}=0$, and $Ax+f=A\hat{x}+\hat{f}$, we have $Ah=f-\hat{f}$ and
\begin{equation}
 \label{eq 3.2.6}
A_{B^c,:} h=(f-\hat{f})_{B^c}=-\hat{f}_{B^c}.
\end{equation}
Then we have the following inequality
\begin{align*}
&\|\hat{x}\|_1+\lambda\|\hat{f}\|_1\\
&=\langle \hat{x}_T, \sgn(\hat{x}_T)\rangle+\|\hat{x}_{T^c}\|_1+\lambda(\langle \hat{f}_B, \sgn(\hat{f}_B)\rangle+\|\hat{f}_{B^c}\|_1)\\
&\geq\langle \hat{x}_T, \sgn(x_T)\rangle+\|\hat{x}_{T^c}\|_1+\lambda(\langle \hat{f}_B, \sgn(f_B)\rangle+\|\hat{f}_{B^c}\|_1)\\
&=\langle x_T+h_T, \sgn(x_T)\rangle+\|h_{T^c}\|_1+\lambda(\langle f_B-A_{B,:}h, \sgn(f_B)\rangle+\|A_{B^c,:}h\|_1)&~~&\text{By (\ref{eq 3.2.5}), (\ref{eq 3.2.6})}\\
&=\|x\|_1+\lambda\|f\|_1+\|h_{T^c}\|_1+\lambda\|A_{B^c,:}h\|_1+\langle h_T, \sgn(x_T)\rangle-\lambda\langle A_{B,:}h, \sgn(f_B)\rangle.
\end{align*}
Since $\|\hat{x}\|_1+\lambda\|\hat{f}\|_1\leq \|x\|_1+\lambda\|f\|_1$, we have
\begin{equation}
 \label{ineq 3.2.7}
\|h_{T^c}\|_1+\lambda\|A_{B^c,:}h\|_1+\langle h_T,
\sgn(x_T)\rangle-\lambda\langle A_{B,:}h, \sgn(f_B)\rangle\leq 0.
\end{equation}
By (\ref{eq 3.2.4}), we have
\[
\langle h_T, v_T\rangle+\langle h_{T^c},  v_{T^c}\rangle
=\langle h, v\rangle
=\langle h, A_{B^c,:}^*q_{B^c}+A_{B,:}^*\lambda\sgn{f_B}\rangle
=\langle A_{B^c,:}h, q_{B^c}\rangle+\lambda\langle A_{B,:}h,\sgn{f_B}\rangle,
\]
and then by (\ref{ineq 3.2.3}),
\begin{align*}
\langle h_T, \sgn(x_T)\rangle-\lambda\langle A_{B,:}h, \sgn(f_B)\rangle
&=\langle h_T, (\sgn(x_T)-v_T)\rangle+\langle A_{B^c,:}h, q_{B^c}\rangle-\langle h_{T^c},  v_{T^c}\rangle\\
&\geq-{{\lambda}\over 4}\|h_T\|_2-{1\over 4}\lambda\|A_{B^c,:}h\|_1-{1\over 4}\|h_{T^c}\|_1.
\end{align*}
Unite it with (\ref{ineq 3.2.7}), we have
\begin{equation}
 \label{ineq 3.2.8}
-{{\lambda}\over 4}\|h_T\|_2+{3\over 4}\lambda\|A_{B^c,:}h\|_1+{3\over
4}\|h_{T^c}\|_1\leq 0.
\end{equation} 
By (\ref{ineq 3.2.1}), we have
$\left\|\sqrt{m\over{m-m_b}}A_{B^c,T}^*\right\|_{2, 2}\leq \sqrt{3\over 2}$ and the smallest
singular value of ${m\over{m-m_b}}A_{B^c,T}^*A_{B^c,T}$ is at least ${1\over
2}$. Therefore,
\begin{align*}
\|h_T\|_2
&\leq2\left\|{{m}\over{m-m_b}}A_{B^c,T}^*A_{B^c,T}h_T\right\|_2\\
&\leq2\left(\left\|{{m}\over{m-m_b}}A_{B^c,T}^*A_{B^c,T^c}h_{T^c}\right\|_2+\left\|{{m}\over{m-m_b}} A_{B^c,T}^*A_{B^c,:}h\right\|_2\right)\\
&\leq2\left\|{{m}\over{m-m_b}}A_{B^c,T}^*A_{B^c,T^c}h_{T^c}\right\|_2+\sqrt{6}\left\|\sqrt{{m} \over {m-m_b}}A_{B^c,:}h\right\|_2\\
&\leq2\sum_{i \in T^c}\left\|{{m}\over{m-m_b}}A_{B^c,T}^*A_{B^c,i}\right\|_2|h_{\{i\}}|+\sqrt{6}\left\|\sqrt{{m} \over{m-m_b}}A_{B^c,:}h\right\|_2&~~&\text{By the triangle inequality}\\
&\leq2\|h_{T^c}\|_1+\sqrt{6}\left\|\sqrt{{m}\over{m-m_b}}A_{B^c,:}h\right\|_1&~~&\text{By (\ref{ineq 3.2.2})}.
\end{align*}
Plugging this into (\ref{ineq 3.2.8}), we have$\left({3\over
4}-{1\over2}\lambda\right)\|h_{T^c}\|_1+\left({3\over 4}-{{\sqrt{6}}\over
4}\sqrt{m\over{m-m_b}}\right)\lambda\|A_{B^c,:}h\|_1\leq 0$. We know ${3\over 4}-{{\sqrt{6}}\over
4}\sqrt{m\over{m-m_b}}>0$ when $\beta$ is sufficiently small. Moreover, by the assumption $\lambda<{3\over 2}$, we have
$h_{T^c}=0$ and $A_{B^c,:}h=0$. Since $A_{B^c,:}h=A_{B^c,T}h_T+A_{B^c,T^c}h_{T^c}$, we have $A_{B^c,T}h_T=0$. The inequality (\ref{ineq 3.2.1}) implies that $A_{B^c,T}$ is
injective, so $h_T=0$ and $h=h_T+h_{T^c}=0$, which implies $(\hat{x}, \hat{f})=(x, f)$.
\end{proof}
Now let's construct a vector $q_{B^c}$ satisfying the requirement (\ref{ineq
3.2.3}) by choosing an appropriate $\lambda$.\\
\\
\begin{proof}(of Theorem \ref{thm 1.2})
Set $\lambda={1\over{\sqrt{\log n}}}$. It suffices to construct a $q_{B^c}$
satisfying \eqref{ineq
3.2.3}. Denoting $u=A_{B^c.:}^*q_{B^c}$, we only need to construct a $q_{B^c}$
satisfying
\begin{eqnarray*}
 \|u_T+\lambda A_{B,T}^*\sgn(f_B)-\sgn(x_T)\|_2\leq {\lambda \over 4},~~
 \|u_{T^c}\|_\infty \leq {1\over 8},~~
\|\lambda A_{B,:}^*\sgn(f_B)\|_\infty \leq {1\over 8},~~
 \|q_{B^c}\|_\infty \leq  {\lambda \over 4}.
\end{eqnarray*}

Now let's construct our $q_{B^c}$ by the golfing scheme. First we have to write
$A_{B^c,:}$ as a block matrix. We divide $B^c$ into $l=\lfloor \log_2
n+1\rfloor=\lfloor{{\log n}\over{\log 2}}+1\rfloor$ disjoint subsets: $B^c=G_1\cup...\cup G_l$
where $|G_i|=m_i$. Then we have $\sum_{i=1}^l m_i=m-m_b$ and
\[
 A_{B^c,:}=\begin{bmatrix} A_{G_1,:} \\ \cdots \\ A_{G_l,:} \end{bmatrix}.
\]
We want to mention that the partition of $B^c$ is deterministic, not depending on $A$, so $A_{G_1,:}, ..., A_{G_l,:}$ are independent. 
Noticing $m_b\leq \beta {m\over \mu}\leq \beta m$, by letting $\beta$
sufficiently small, we can require
\[
{m\over{m_1}}\leq C, ~~{m\over{m_2}}\leq C, ~~{m\over{m_k}}\leq C\log n
\text{~~for~~} k=3,...,l
\]
for some absolute constant $C$. Since $s\leq \alpha {m\over {\mu \log^2 n}}$, we
have
\begin{equation}
\label{ineq 3.2.9}
s\leq \alpha C {{m_1}\over{\mu \log^2 n}}, ~~s\leq \alpha C {{m_2}\over{\mu
\log^2 n}}, ~~s\leq \alpha C {{m_k}\over{\mu \log n}} \text{~~for~~} k=3,...,l.
\end{equation}
Then by Lemma \ref{lemma 3.1}, replacing $A$ with $\sqrt{m\over{m_j}}A_{G_j,T}$,
we have the following inequalities:
\begin{eqnarray}
\label{ineq 3.2.10}
 \left\|{m\over{m_j}}A_{G_j,T}^*A_{G_j,T}-I\right\|_{2, 2}&\leq& {1\over {2\sqrt{\log n}}}
\text{~for~~}j=1,2;\\
\label{ineq 3.2.11}
 \left\|{m\over{m_j}}A_{G_j,T}^*A_{G_j,T}-I\right\|_{2, 2}&\leq& {1\over 2}
\text{~for~~}j=3,...,l;
\end{eqnarray}
with high probability provided $\alpha$ is sufficiently small.\\
\\

Now let's give an explicit construction of $q_{B^c}$. Define
\begin{equation}
\label{eq 3.2.12}
p_0=\sgn(x_T)-\lambda A_{B,T}^*\sgn(f_B)
\end{equation} 
and
\begin{equation}
\label{eq 3.2.13} 
p_i=\left(I-{m\over{m_i}}A_{G_i,T}^*A_{G_i,T}\right)p_{i-1}=\left(I-{m\over{m_i}}A_{G_i,T}^*A_{G_i,T}\right)\cdots\left(I-{m
\over{m_1}}A_{G_1,T}^*A_{G_1,T}\right)p_0 
\end{equation}
for $i=1,...,l$, and construct
\begin{equation}
\label{eq 3.2.14}
 q_{B^c}=\begin{bmatrix} {m\over{m_1}}A_{G_1,T}p_0 \\ \vdots \\{m\over{m_l}}A_{G_l,T}p_{l-1} \end{bmatrix}.
\end{equation}
Then by $u=A_{B^c,:}^*q_{B^c}$, we have
\begin{equation}
\label{eq 3.2.15}
u=A_{B^c,:}^*\begin{bmatrix} {m\over{m_1}}A_{G_1,T}p_0 \\ \vdots \\{m\over{m_l}}A_{G_l,T}p_{l-1} \end{bmatrix}=\sum_{i=1}^l{m
\over{m_i}}A_{G_i,:}^*A_{G_i,T}p_{i-1}.
\end{equation}

We now bound the $\ell_2$ norm of $p_i$. Actually, by (\ref{ineq 3.2.10}), (\ref{ineq 3.2.11}) and (\ref{eq 3.2.13}), we have
\begin{eqnarray}
 \label{ineq 3.2.16}
\|p_1\|_2&\leq& {1\over{2\sqrt{\log n}}}\|p_0\|_2,\\
\label{ineq 3.2.17}
\|p_2\|_2&\leq& {1\over{4\log n}}\|p_0\|_2,\\
\label{ineq 3.2.18}
\|p_j\|_2&\leq& {1\over{\log n}}({1\over 2})^j\|p_0\|_2 \text{~for~~}j=3,...,l.
\end{eqnarray}

Now we will prove our constructed $q_{B^c}$ satisfies the desired requirements:
\paragraph{The proof of $\left\|\lambda A_{B,:}^*\sgn(f_B)\right\|_\infty \leq {1\over 8}$}~\\
By Hoeffding's inequality, for any $i=1,...,n$, we have $\mathbb{P}\left(\left|A_{B,i}^*\sgn(f_B)\right|\geq t\right)\leq 2\exp\left(-{{2t^2}\over{4\left\|A_{B,i}\right\|_2^2}}\right)$. By choosing $t=C\sqrt{\log n}\left\|A_{B,i}\right\|_2$ ($C$ is some absolute constant), with high probability, we have $\left|\lambda A_{B,i}^*\sgn(f_B)\right|\leq \lambda C \sqrt{\log n}\left\|A_{B,i}\right\|_2\leq C\sqrt{{\mu m_b}\over m}\leq \sqrt{\beta}\leq {1\over 8}$, provided $\beta$ is sufficiently small, and this implies $\left\|\lambda A_{B,:}^*\sgn(f_B)\right\|_\infty \leq{1\over 8}$.

\paragraph{The proof of $ \left\|u_T+\lambda A_{B,T}^*\sgn(f_B)-\sgn(x_T)\right\|_2\leq {\lambda \over 4}$}~\\
By (\ref{eq 3.2.15}) and (\ref{eq 3.2.13}), we have
$u_T=\displaystyle\sum\limits_{i=1}^l{m\over{m_i}}A_{G_i,T}^*A_{G_i,T}p_{i-1}=\displaystyle\sum\limits_{i=1}^l (p_{i-1}-p_i)=p_0-p_l$.
Then by (\ref{eq 3.2.12}) we have $\left\|u_T+\lambda A_{B,T}^*\sgn(f_B)-\sgn(x_T)\right\|_2=\left\|u_T-p_0\right\|_2=\|p_l\|_2$. Since
$\left\|\lambda A_{B,:}^*\sgn(f_B)\right\|_\infty\leq 1/8$, we have $\left\|\lambda
A_{B,T}^*\sgn(f_B)\right\|_2\leq {1\over 8}\sqrt{s}$, which implies 
\begin{equation}
\label{ineq 3.2.19}
\|p_0\|_2=\left\|\lambda A_{B,T}^*\sgn(f_B)-\sgn(x_T)\right\|_2\leq {9\over 8}\sqrt{s}.
\end{equation}
 Then by (\ref{ineq 3.2.18}) and $l=\lfloor \log_2 n+1 \rfloor$, we have $\|p_l\|_2\leq {1\over{\log n}}({1\over 2})^l{9\over 8}\sqrt{s}\leq \left({1\over {\log n}}\right)\left({1\over n}\right)\left( {9\over 8}\right)\sqrt{{\alpha m}\over{\mu \log^2 n}}\leq {1\over{4\sqrt{\log n}}}={\lambda \over 4}$, provided $\alpha$ is sufficiently small.

\paragraph{The proof of $\left\|u_{T^c}\right\|_\infty\leq 1/8$}~\\
By (\ref{eq 3.2.15}), we have $u_{T^c}=\displaystyle\sum\limits_{i=1}^l{m\over{m_i}}A_{G_i,T^c}^*A_{G_i,T}p_{i-1}$.
Recall that $A_{G_1,:},..., A_{G_l,:}$ are independent, so by the construction of  $p_{i-1}$ we know $A_{G_i,:}$ and $p_{i-1}$ are independent. Replacing $A$ with $\sqrt{m\over{m_i}}A_{G_i,:}$ in Lemma \ref{lemma 3.2}, and by the sparsity condition (\ref{ineq 3.2.9}), we have  $\displaystyle\sum\limits_{i=1}^l\left\|{m\over{m_i}}A_{G_i,T^c}^*A_{G_i,T}p_{i-1}\right\|_\infty\leq \displaystyle\sum\limits_{i=1}^l {1\over 20}{1\over{\sqrt{s}}}\|p_{i-1}\|_2$ with high probability, provided $\alpha$ is sufficiently small.
By (\ref{ineq 3.2.16}), (\ref{ineq 3.2.17}), (\ref{ineq 3.2.18}) and (\ref{ineq 3.2.19}), we have
$\|u_{T^c}\|_\infty\leq \displaystyle\sum\limits_{i=1}^l {1\over 20}{1\over{\sqrt{s}}}\|p_{i-1}\|_2\leq {1\over
20}{1\over{\sqrt{s}}}2\|p_0\|_2< {1\over 8}$.

\paragraph{The proof of $\left\|q_{B^c}\right\|_\infty\leq {\lambda\over 4}$}~\\
For $k=1,..,l$, we denote $A_{G_k,:}={1\over {\sqrt{m}}}\begin{bmatrix}a_{k_1}^*\\...\\ a_{k_{m_k}}^* \end{bmatrix}$, and $A_{B,:}={1\over {\sqrt{m}}}\begin{bmatrix}\tilde{a}_{1}^*\\...\\ \tilde{a}_{m_b}^* \end{bmatrix}$.
By (\ref{eq 3.2.13}), (\ref{eq 3.2.14}) and (\ref{eq 3.2.12}), it suffices to show that for any $1\leq k\leq l$ and $1\leq j \leq m_k$,
\[
\left|{{\sqrt{m}}\over {m_k}}(a_{k_j})_T^*\left(I-{m\over
{m_{k-1}}}A_{G_{k-1},T}^*A_{G_{k-1},T}\right)\cdots\left(I-{m\over
{m_1}}A_{G_1,T}^*A_{G_1,T}\right)\left(\sgn(x_T)-\lambda A_{B,T}^*\sgn(f_B)\right)\right| \leq {\lambda \over 4}.
\]
Set 
\begin{equation}
\label{eq 3.2.20}
w=\left(I-{m\over{m_1}}A_{G_1,T}^*A_{G_1,T}\right)\cdots\left(I-{m\over{m_{k-1}}}A_{G_{k-1},T}^*A_{G_{k-1},T}\right)(a_{k_j})_T.
\end{equation}
Then it suffices to prove
\[
\left|{{\sqrt{m}}\over {m_k}}w^*\left(\sgn(x_T)-\lambda A_{B,T}^*\sgn(f_B)\right)\right| \leq {\lambda \over 4}.
\]
Since $w$ and $\sgn(x_T)$ are independent, by Hoeffding's inequality and conditioning on w, we have $\mathbb{P}\left(\left|w^*\sgn(x_T)\right|\geq t\right)\leq 2 \exp\left(-{{2t^2}\over{4\|w\|_2^2}}\right)$ for any $t>0$. Then with high probability we have 
\begin{equation}
\label{ineq 3.2.21}
\left|w^*\sgn(x_T)\right|\leq C\sqrt{\log n}\|w\|_2
\end{equation} 
for some absolute constant $C$.\\
\\
Setting $z=\sgn(f_B)$, we have $w^*A_{B,T}^*\sgn(f_B)={1\over {\sqrt{m}}}\displaystyle\sum\limits_{i=1}^{m_b}\left[(\tilde{a}_i)_T^*w\right] z_{\{i\}}$. Since $w$, $A_{B,T}$ and $z$ are independent, conditioning on $w$ we have
\[
\mathbb{E}\{[(\tilde{a}_i)_T^*w]z_{\{i\}}\}=\mathbb{E}\{(\tilde{a}_i)_T^*w\} \mathbb{E}\{z(i)\}=0,
\]
\[
\left|[(\tilde{a}_i)_T^*w] z_{\{i\}}\right|\leq \|w\|_2\left\|(\tilde{a}_i)_T\right\|_2\leq \sqrt{s\mu}\|w\|_2\leq \sqrt{{\alpha m}\over{\log^2 n}}\|w\|_2,
\]
and
\[
\mathbb{E}\{ \left|[(\tilde{a}_i)_T^*w] z_{\{i\}}\right|^2 \}=\mathbb{E}\{ [w^*(\tilde{a}_i)_T][(\tilde{a}_i)_T^*w] \}
=w^*\mathbb{E}\{(\tilde{a}_i)_T(\tilde{a}_i)_T^* \}w=\|w\|_2^2.
\]
By Bernstein's inequality, we have 
\[
\mathbb{P}\left(\left|w^*A_{B,T}^*\sgn(f_B)\right|\geq {t\over {\sqrt{m}}}\right)\leq 2\exp\left(-{{t^2/2}\over{m_b\|w\|_2^2+\sqrt{{\alpha m}\over{\log^2 n}}\|w\|_2 t/3}}\right).
\]
By choosing some numerical constant $C$ and $t=C\sqrt{m\log n}\|w\|_2$, we have 
\begin{equation}
\label{ineq 3.2.22}
\left|w^*A_{B,T}^*\sgn(f_B)\right|\leq C\sqrt{\log n}\|w\|_2
\end{equation}
with high probability, provided $\alpha$ is sufficiently small. \\
\\
By (\ref{ineq 3.2.21}) and (\ref{ineq 3.2.22}), we have
\begin{equation}
\label{ineq 3.2.23}
\left|{{\sqrt{m}}\over {m_k}}w^*\left(\sgn(x_T)-\lambda A_{B,T}^*\sgn(f_B)\right)\right| \leq {{\sqrt{m}}\over {m_k}}C \sqrt{\log n}\|w\|_2,
\end{equation}
for some numerical constant $C$.\\
\\
When $k\geq 3$, by (\ref{eq 3.2.20}), (\ref{ineq 3.2.10}) and (\ref{ineq 3.2.11}), we have $\|w\|_2\leq ({1\over 2})^{k-1}{1\over{\log n}}\sqrt{\mu s}\leq{\sqrt{\alpha m}\over{\log^2 n}}$. Recalling ${m\over {m_k}}\leq C\log n$, by (\ref{ineq 3.2.23}), we have $\left|{{\sqrt{m}}\over {m_k}}w^*\left(\sgn(x_T)-\lambda A_{B,T}^*\sgn(f_B)\right)\right| \leq C\left({m\over {m_k}}\right)\sqrt{\alpha}(\log n)^{-3/2}\leq {\lambda \over 4}$ provided $\alpha$ is sufficiently small.\\
\\ 
When $k\leq 2$, by (\ref{eq 3.2.20}) and (\ref{ineq 3.2.10}), we have $\|w\|_2\leq \sqrt{\mu s}\leq {\sqrt{\alpha m}\over {\log n}}$. Recalling ${m\over {m_k}}\leq C$, by (\ref{ineq 3.2.23}), we have $\left|{{\sqrt{m}}\over {m_k}}w^*\left(\sgn(x_T)-\lambda A_{B,T}^*\sgn(f_B)\right)\right| \leq C\left({m\over {m_k}}\right)\sqrt{\alpha}(\log n)^{-1/2}\leq {\lambda \over 4}$ provided $\alpha$ is sufficiently small. 
\end{proof}
Here we would like to compare our golfing scheme with that in \cite{CP10}. There are mainly two differences. One is that we have an extra term $\lambda A_{B,:}^*\sgn(f_B)$ in the dual vector. To obtain the inequality $\|v_{T^c}\|_\infty\leq 1/4$, we propose to bound  $\|u_{T^c}\|_\infty$ and $\|\lambda A_{B,:}^*\sgn(f_B)\|_\infty$ respectively, and this will lead to the extra log factor compared with \cite{CP10}. Moreover, by using the golfing scheme to  construct the dual vector, we need to bound the term $\|q_{B^c}\|_\infty$, which is not necessary in \cite{CP10}. This inevitably incurs the random signs assumptions of the signal.

\section{A Proof of Theorem \ref{thm 1.3}}
In this section, the capital letters $X$, $Y$ etc represent matrices, and the symbols in script font $\mathcal{I}$, $\mathcal{P}_T$, etc represent linear operators from a matrix space to a matrix space. Moreover, for any $\Omega_0 \subset [n] \times [n]$ we have $\PP_{\Omega_0}M$ is to keep the entries of $M$ on the support $\Omega_0$ and to change other entries into zeros. For any $n \times n$ matrix $A$, denote by $\|A\|_F$, $\|A\|$, $\|A\|_\infty$ and $\|A\|_*$ respectively the Frobenius norm, operator norm (the largest singular value), the biggest magnitude of all elements, and  the nuclear norm(the sum of all singular values).\\
\\
Similarly to Section 3, instead of denoting them as $C_1$, $C_2$, ..., we just use $C$, whose values change from line to line. Also, we will use the phrase ``with high probability" to mean with probability at least $1-Cn^{-c}$, where $C>0$ is a numerical constant and $c=3, 4, \text{~or~}5$ depending on the context.\\

\subsection{A model equivalent to Model 3.1}
Model 3.1 is natural and used in \cite{CLMW}, but we will use the following equivalent model for the convenience of proof:

\paragraph{Model 3.2:}
1. Fix an $n$ by $n$ matrix $K$, whose entries are either $1$ or $-1$. \\
2. Define two independent random subsets of $[n] \times [n]$: $\Gamma' \sim \text{Ber}((1-2s)\rho)$ and $\Omega' \sim \text{Ber}({{2s\rho}\over{1-\rho+2s\rho}})$. Moreover, let $O:=\Gamma' \cup \Omega'$, which thus satisfies $O \sim \text{Ber}(\rho)$.\\
3. Define an $n \times n$ random matrix $W$ with independent entries $W_{ij}$ satisfying $\P(W_{ij}=1)=\P(W_{ij}=-1)={1\over 2}$.\\
4. Define $\Omega'' \subset \Omega'$: $\Omega'':=\{(i, j): (i, j)\in \Omega', W_{ij}=K_{ij}\}$.\\
5. Define $\Omega:=\Omega''/\Gamma'$, and $\Gamma:=O/\Omega$.\\
6. Let $S$ satisfy $\sgn(S):=\PP_{\Omega}(K)$.\\
\\
Obviously, in both Model 3.1 and Model 3.2 the whole setting is deterministic if we fix $(O, \Omega)$. Therefore, the probability of $(\hat{L}, \hat{S})=(L, S)$ is determined by the joint distribution of $(O, \Omega)$. It is not difficult to prove that the joint distributions of $(O, \Omega)$ in both models are the same. Indeed, in Model 3.1, we have that $(1_{\{(i, j)\in O\}}, 1_{\{(i, j)\in \Omega\}})$ are iid random vectors with the probability distribution $\P(1_{\{(i, j)\in O\}}=1)=\rho$, $\P(1_{\{(i, j)\in \Omega\}}=1|1_{\{(i, j)\in O\}}=1)=s$ and $\P(1_{\{(i, j)\in \Omega\}}=1|1_{\{(i, j)\in O\}}=0)=0$. In Model 3.2, we have 
\[
(1_{\{(i, j)\in O\}}, 1_{\{(i, j)\in \Omega\}})=(\max(1_{\{(i, j)\in \Gamma'\}}, 1_{\{(i, j)\in \Omega'\}}), 1_{\{(i, j)\in \Omega'\}}1_{\{W_{i,j}=K_{i,j}\}}1_{\{(i, j)\in \Gamma'^c\}}).
\]
This implies that $(1_{\{(i, j)\in O\}}, 1_{\{(i, j)\in \Omega\}})$ are independent random vectors. Moreover, it is easy to calculate that $\P(1_{\{(i, j)\in O\}}=1)=\rho$, $\P(1_{\{(i, j)\in \Omega\}}=1)=s\rho$ and $\P(1_{\{(i, j)\in \Omega\}}=1, 1_{\{(i, j)\in O\}}=0)=0$. Then we have 
\[
\P(1_{\{(i, j)\in \Omega\}}=1|1_{\{(i, j)\in O\}}=1)=\P(1_{\{(i, j)\in \Omega\}}=1, 1_{\{(i, j)\in O\}}=1)/\P(1_{\{(i, j)\in O\}}=1)=s,
\]
and
\[
\P(1_{\{(i, j)\in \Omega\}}=1|1_{\{(i, j)\in O\}}=0)=\P(1_{\{(i, j)\in \Omega\}}=1, 1_{\{(i, j)\in O\}}=0)/\P(1_{\{(i, j)\in O\}}=0)=0.
\]
Notice that although $(1_{\{(i, j)\in O\}}, 1_{\{(i, j)\in \Omega\}})$ depends on $K$, its distribution does not. By the above we know that $(O, \Omega)$ has the same distribution in both models. Therefore in the following we will use Model 3.2 instead. The advantage of using Model 3.2 is that we can utilize $\Gamma'$, $\Omega'$, $W$, etc. as auxiliaries.\\
\\
In the next section we prove some supporting lemmas which are useful for the proof of the main theorem.
\subsection{Supporting lemmas}
Define  $T:=\{UX^*+YV^*, X,Y \in \mathbb{R}^{n \times r}\}$ a subspace of $\mathbb{R}^{n \times n}$. Then the orthogonal projectors $\mathcal{P}_T$ and $\mathcal{P}_{T^\perp}$ in $\mathbb{R}^{n \times n}$   satisfy $\PT X=UU^*X+XVV^*-UU^*XVV^*$ and $\PTp X= (I-UU^*)X(I-VV^*)$ for any $X \in \R^{n\times n}$. This means $\|\PTp X\|\leq \|X\|$ for any
$X$. Recalling the incoherence conditions: for any $i \in \{1, ..., n\}$,
$\|UU^*e_i\|^2\leq {{\mu r}\over n}$ and $\left\|VV^*e_i\right\|^2\leq {{\mu r}\over n}$, we have  $\|\PT (e_ie_j^*)\|_\infty\leq {{2\mu r}\over n}$ and  $\|\PT (e_ie_j^*)\|_F \leq \sqrt{{2\mu r}\over n}$ \cite{CR09, CT10}.

\begin{lemma}(Theorem 4.1 of \cite{CR09})
\label{lemma 4.1} 
Suppose $\Omega_0 \sim \text{Ber}(\rho_0)$. Then with high probability, $\|\PT - \rho_0^{-1} \PT \POzero \PT \| \le \epsilon$,  provided that $\rho_0 \ge C_0 \, \epsilon^{-2} \, \frac{\mu r \log n}{n}$ for some numerical constant $C_0 > 0$.
\end{lemma}

The original idea of the proof of this theorem is due to \cite{Rudelson99}.

\begin{lemma}(Theorem 3.1 of \cite{CLMW})
\label{lemma 4.2} 
  Suppose $Z \in \text{Range} (\PT)$ is a fixed matrix, $\Omega_0 \sim \text{Ber}(\rho_0)$, and $\epsilon \leq 1$ is an arbitrary constant. Then with high probability $\|(\I - \rho_0^{-1}\PT\POzero) Z\|_\infty \le \epsilon \|Z\|_\infty$ provided that $\rho_0 \ge C_0 \, \epsilon^{-2} \, \frac{\mu r \log n}{n}$  for some numerical constant $C_0 > 0$.
\end{lemma}

\begin{lemma}(Theorem 6.3 of \cite{CR09})
\label{lemma 4.3}
Suppose $Z$ is a fixed matrix, and $\Omega_0 \sim \text{Ber}(\rho_0)$. Then
with high probability, $\|(\rho_0\I - \POzero) Z\| \le C_0' \sqrt{n p\log n} \|Z\|_\infty$
provided that $\rho_0\leq p$ and $p \geq \, C_0\frac{\log n}{n}$ for some numerical constants $C_0 > 0$ and $C_0'>0$.
\end{lemma}

Notice that we only have $\rho_0=p$ in Theorem 6.3 of \cite{CR09}. By a very slight modification in the proof (specifically, the proof of Lemma 6.2) we can have $\rho_0\leq p$ as stated above. 

\subsection{A proof of Theorem 1.3}
By Lemma 3.1, we have we have $\|{1\over{(1-2s)\rho}}\PP_T \PP_{\Gamma'}\PP_T-\PP_T\|\leq {1\over 2}$ and $\|{1\over{\sqrt{(1-2s)\rho}}}\PP_T\PP_{\Gamma'}\|\leq \sqrt{3/2}$ with high probability provided $C_\rho$ is sufficiently large and $C_s$ is sufficiently small. We will assume both inequalities hold all through the paper.
\begin{theorem}
\label{teo:kktdg}
If there exists an $n \times n$ matrix $Y$ obeying
\begin{equation}
\label{eq:dual-certif1}
\begin{cases}
  \|\PT Y+\PT (\lambda \PP_{O/\Gamma'}W-UV^*)\|_F\le {\lambda\over{n^2}},\\
  \|\PTp Y+\PTp (\lambda \PP_{O/\Gamma'}W)\| \leq {1\over 4},\\
   \PP_{\Gamma'^c}Y=0,\\
  \|\PP_{\Gamma'} Y\|_\infty \le {\lambda \over 4},
\end{cases}
\end{equation}
where $\lambda={1\over{\sqrt{n\rho \log n}}}$. Then the solution $(\hat{L}, \hat{S})$ to \eqref{decoding 3} satisfies $(\hat{L}, \hat{S})=(L, S)$.
\end{theorem}

\begin{proof}
Set $H=\hat{L}-L$. The condition $\PP_O(L)+S=\PP_O(\hat{L})+\hat{S}$ implies that $\PP_O(H)=S-\hat{S}$. Then $\hat{S}$ is supported on $O$ because $S$ is supported on $\Omega \subset O$. By considering the subgradient of the nuclear norm at $L$, we have 
\[
\|\hat{L}\|_*\geq \|L\|_*+\langle \PP_T H, UV^*\rangle+\|\PP_{T^{\perp}}H\|_*.
\]
By the definition of $(\hat{L}, \hat{S})$, we have
\[
\|\hat{L}\|_*+\lambda\|\hat{S}\|_1\leq \|L\|_*+\lambda\|S\|_1.
\]
By the two inequalities above, we have
\[
\lambda\|S\|_1-\lambda\|\hat{S}\|_1\geq \langle \PP_T(H), UV^*\rangle+\|\PP_{T^{\perp}}H\|_*,
\]
which implies
\[
\lambda\|S\|_1-\lambda\|\PP_{O/\Gamma'}(\hat{S})\|_1\geq \langle H, UV^* \rangle+\|\PP_{T^\perp}(H)\|_*+\lambda\|\PP_{\Gamma'}(\hat{S})\|_1.
\]
On the other hand, 
\begin{align*}
\|\PP_{O/\Gamma'}\hat{S}\|_1
&=\|S+\PP_{O/\Gamma'}(-H)\|_1\\
&\geq\|S\|_1+\langle \sgn(S), \PP_{\Omega}(-H)\rangle+\|\PP_{O/(\Gamma'\cup\Omega)}(-H)\|_1\\
&\geq\|S\|_1+\langle \PP_{O/\Gamma'}(W), -H\rangle.
\end{align*}
By the two inequalities above and the fact $\PP_{\Gamma'}\hat{S}=\PP_{\Gamma'}(\hat{S}-S)=-\PP_{\Gamma'}H$, we have 
\begin{equation}
\label{4.2}
\|\PP_{T^\perp} (H)\|_*+\lambda\|\PP_{\Gamma'}(H)\|_1\leq \langle H, \lambda \PP_{O/\Gamma'}(W)-UV^*\rangle.
\end{equation}
By the assumptions of $Y$, we have
\begin{align*}
&\langle H, \lambda \PP_{O/\Gamma'}(W)-UV^*\rangle\\
&=\langle H, Y+\lambda \PP_{O/\Gamma'}(W)-UV^*\rangle-\langle H, Y\rangle\\
&=\langle \PP_T(H), \PP_T(Y+\lambda \PP_{O/\Gamma'}(W)-UV^*)\rangle+\langle \PP_{T^{\perp}}(H), \PP_{T^{\perp}}(Y+\lambda\PP_{O/\Gamma'}(W))\rangle\\
&-\langle \PP_{\Gamma'}(H), \PP_{\Gamma'}(Y)\rangle-\langle \PP_{\Gamma'^c}(H), \PP_{\Gamma'^c}(Y)\rangle\\
&\leq {\lambda \over {n^2}}\|\PP_T(H)\|_F+{1\over 4}\|\PP_{T^{\perp}}(H)\|_*+{\lambda \over 4}\|\PP_{\Gamma'}(H)\|_1.
\end{align*}
By inequality \ref{4.2}, 
\begin{equation}
\label{4.3}
{3\over 4}\|\PP_{T^\perp}(H)\|_*+{{3\lambda} \over 4}\|\PP_{\Gamma'}(H)\|_1\leq {\lambda \over {n^2}}\|\PP_{T}(H)\|_F.
\end{equation}
Recall that we assume $\|{1\over{(1-2s)\rho}}\PP_T \PP_{\Gamma'}\PP_T-\PP_T\|\leq {1\over 2}$ and $\|{1\over{\sqrt{(1-2s)\rho}}}\PP_T\PP_{\Gamma'}\|\leq \sqrt{3/2}$ all through the paper. Then
\begin{align*}
\|\PP_T(H)\|_F&\leq 2\|{1\over{(1-2s)\rho}}\PP_T \PP_{\Gamma'}\PP_T(H)\|_F\\
&\leq 2\|{1\over{(1-2s)\rho}}\PP_T \PP_{\Gamma'}\PP_{T^{\perp}}(H)\|_F+2\|{1\over{(1-2s)\rho}}\PP_T \PP_{\Gamma'}(H)\|_F\\
&\leq \sqrt{6\over{(1-2s)\rho}}\|\PP_{T^{\perp}}H\|_F+\sqrt{6\over{(1-2s)\rho}}\|\PP_{\Gamma'}H\|_F.
\end{align*}
By inequality \ref{4.3}, we have
\[
({3\over 4}-{\lambda\over{n^2}}\sqrt{6\over{(1-2s)\rho}})\|\PP_{T^{\perp}}(H)\|_F+({3\lambda\over 4}-{\lambda\over{n^2}}\sqrt{6\over{(1-2s)\rho}})\|\PP_{\Gamma'}H\|_F\leq 0.
\]
Then $\PP_{T^{\perp}}(H)=\PP_{\Gamma'}H=0$, which implies $\PP_{\Gamma'}\PP_{T}(H)=0$. Since $\PP_{\Gamma'}\PP_T$ is injective ($\|{1\over{(1-2s)\rho}}\PP_T \PP_{\Gamma'}\PP_T-\PP_T\|\leq {1\over 2}$) on $T$, we have $\PP_T(H)=0$. Then we have $H=0$.
\end{proof}

Suppose we can construct $Y$ and $\widetilde{Y}$ satisfying
\begin{equation}
\label{eq:dual-certif2}
\begin{cases}
  \|\PT Y+\PT (\lambda \PP_{\Omega'}W-UV^*)\|_F\le {\lambda\over{2n^2}},\\
  \|\PTp Y+\PTp (\lambda \PP_{\Omega'}W)\| \leq {1\over 4},\\
   \PP_{\Gamma'^c}Y=0,\\
  \|\PP_{\Gamma'} Y\|_\infty \le {\lambda \over 4}. 
\end{cases}
\end{equation}
and
\begin{equation}
\label{eq:dual-certif3}
\begin{cases}
  \|\PT \widetilde{Y}+\PT (\lambda (2\PP_{\Omega'/\Gamma'}(W)-\PP_{\Omega'}W)-UV^*)\|_F\le {\lambda\over{2n^2}},\\
  \|\PTp \widetilde{Y}+\PTp (\lambda (2\PP_{\Omega'/\Gamma'}(W)-\PP_{\Omega'}W))\| \leq {1\over 4},\\
   \PP_{\Gamma'^c}\widetilde{Y}=0,\\
  \|\PP_{\Gamma'} \widetilde{Y}\|_\infty \le {\lambda \over 4}. 
\end{cases}
\end{equation}
Then $\overline{Y}=(Y+\tilde{Y})/2$ will satisfy \ref{eq:dual-certif1}. By the assumptions in Model 2, $(\Gamma', \PP_{\Omega'}W)$ and $(\Gamma', 2\PP_{\Omega'/\Gamma'}(W)-\PP_{\Omega'}W)$ have the same distribution. Therefore, if we can construct $Y$ satisfying \eqref{eq:dual-certif2} with high probability, we can also construct $\widetilde{Y}$ satisfying 
\eqref{eq:dual-certif3} with high probability. Therefore to prove Theorem 1.3, we only need to prove that there exists $Y$ satisfying \eqref{eq:dual-certif2} with high probability:\\
\\
\begin{proof}(of Theorem \ref{thm 1.3})
Notice that $\Gamma' \sim \text{Ber}((1-2s)\rho)$. Suppose that $q$ satisfying $1-(1-2s)\rho=(1-{{(1-2s)\rho}\over 6})^2(1-q)^{l-2}$, where $l=\lfloor 5\log n+1 \rfloor$. This implies that $q\geq C\rho/\log(n)$. Define $q_1=q_2=(1-2s)\rho/6$ and $q_3=...=q_l=q$.
Then in distribution we can let $\Gamma'=\Gamma_1\cup...\cup\Gamma_l$, where $\Gamma_j \sim \text{Ber}(q_j)$ independently. \\
Construct
\begin{equation*}
\begin{cases} 
  Z_0=\PP_T(UV^*-\lambda \PP_{\Omega'}W),\\
  Z_j=(\PT-{1\over q_j}\PT \PGj \PT)Z_{j-1} \text{ for } j=1,..., j_0.,\\
  Y=\sum_{j=1}^l{1\over q_j}\PGj Z_{j-1},\\
\end{cases}
\end{equation*}
Then by Lemma \ref{lemma 4.1}, we have 
\[
\|Z_j\|_F \leq {1\over 2} \|Z_{j-1}\|_F \text{ for } j=1,...,l.
\]
with high probability provided $C_\rho$ is large enough and $C_s$ is small enough. Then $\|Z_j\|_F \leq ({1\over 2})^j\|Z_0\|_F$. By the construction of $Z_j$, we know that $Z_j \in \text{Range} (\PT)$ and $Z_j=(I-{1\over q}\PT\PGj)Z_{j-1}$. Then similarly, by Lemma \ref{lemma 4.2}, we have 
\[
\|Z_1\|_\infty \leq {1\over {2\sqrt{\log n}}} \|Z_0\|_\infty,
\]
and
\[
\|Z_j\|_\infty  \leq {1\over {2^j \log n}}\|Z_0\|_\infty \text{~~for~~}j=2,...,l
\]
with high probability provided $C_\rho$ is large enough and $C_s$ is small enough. Also, by Lemma \ref{lemma 4.3} we have 
\[
\|(I - {1\over q} \PGj) Z_{j-1}\| \leq C \sqrt{\frac{n\log n}{q}} \|Z_{j-1}\|_\infty \text{~~for~~} j=1,..., l
\]
with high probability provided $C_\rho$ is large enough and $C_s$ is small enough.\\
\\
We first bound $\|Z_0\|_F$ and $\|Z_0\|_\infty$. Obviously $\|Z_0\|_\infty\leq \|UV^*\|_\infty+\lambda \|\PT\PP_{\Omega'}(W)\|_\infty$. Recall that for any $i, j \in [n]$, we have $\|\PT (e_ie_j^*)\|_\infty\leq {{2\mu r}\over n}$ and  $\|\PT (e_ie_j^*)\|_F \leq \sqrt{{2\mu r}\over n}$. Moreover, $\PP_{\Omega'}(W)$ satisfies $(\PP_{\Omega'}(W))_{;}$ are iid random variables with the distribution
\[
(\PP_{\Omega'}(W))_{ij}=\left\{ \begin{array}{rcl} 1 & \mbox{with probability} & {{s\rho}\over{1-\rho+2s\rho}}\\ 0 & \mbox{with probability} & {{1-\rho}\over{1-\rho+2s\rho}} \\ -1 & \mbox{with probability} &  {{s\rho}\over{1-\rho+2s\rho}} \end{array}\right.
\]
Then by Bernstein's inequality, we have
\begin{align*}
\mathbb{P}\left(\left|\left\<\PT(\PP_{\Omega'}(W)), e_ie_j^*\right\>\right|\geq t\right)&=\mathbb{P}\left(\left|\left\<\PP_{\Omega'}(W), \PT(e_ie_j^*)\right\>\right|\geq t\right)\\
&\leq 2\exp(-{{t^2/2}\over{\sum EX_j^2+Mt/3}}),
\end{align*}
where we have
\[
\sum EX_j^2= {{2s\rho}\over{1-\rho+2s\rho}} \|\PT e_ie_j^*\|_F^2\leq C\rho s {{\mu r}\over {n}},
\]
and
\[
M=\|\PT e_ie_j^*\|_\infty \leq {{2\mu r}\over n}.
\]
Then with high probability we have $\|\PT \PP_{\Omega'}(W)\|_\infty\leq C\sqrt{\rho{{\mu r \log n }\over n}}(\geq C\sqrt{C_\rho{{\mu r \log^2 n}\over n}{{\mu r \log n }\over n}}>C\sqrt{C_\rho}M\log n)$. Then by $\|UV^*\|_\infty\leq {\sqrt{\mu r}\over n}$ we have $\|Z_0\|_\infty\leq C{\sqrt{\mu r}\over n}$, which implies $\|Z_0\|_F \leq  n\|Z_0\|_\infty \leq C{\sqrt{ \mu r}}$ .\\
\\
Now we want to prove $Y$ satisfies $\ref{eq:dual-certif2}$ with high probability. Obviously $\PP_{\Gamma'^c}Y=0$. It suffices to prove
\begin{equation}
\label{eq:dual-certif4}
\begin{cases}
  \|\PT Y+\PT (\lambda \PP_{\Omega'}(W)-UV^*)\|_F\le {\lambda\over{2n^2}},\\
  \|\PTp Y\| \leq {1\over 8},\\
  \|\PTp (\lambda \PP_{\Omega'}(W))\| \leq {1\over 8},\\
  \|\PP_{\Gamma'} Y\|_\infty \le {\lambda \over 4}. 
\end{cases}
\end{equation}
First,
\begin{eqnarray*}
\|\PT Y+\PT (\lambda \PP_{\Omega'}(W)-UV^*)\|_F &=& \|Z_0-(\sum_{j=1}^{l}{1\over q_j}\PT \PGj Z_{j-1})\|_F\\
                                   &=&\|\PT Z_0-(\sum_{j=1}^{l}{1\over q_j}\PT \PGj \PT Z_{j-1})\|_F\\
                                   &=&\|(\PT-{1\over q_1}\PT\PP_{\Gamma_1}\PT)Z_0-(\sum_{j=2}^{j_0}{1\over q_j}\PT \PGj \PT Z_{j-1})\|_F\\
                                   &=&\|\PT Z_1-(\sum_{j=2}^{l}{1\over q_j}\PT \PGj \PT Z_{j-1})\|_F\\
                                   &=&...=\|Z_{l}\|_F\ \leq C({1\over 2})^l\sqrt{\mu r}\leq{\lambda \over {n^2}}.
\end{eqnarray*}
Second,
\begin{eqnarray*}
\|\PTp Y\| &= & \|\PTp \sum_{j=1}^l{1\over q_j}\PGj Z_{j-1}\|\\
                      &\leq& \sum_{j=1}^l\|{1\over q_j}\PTp\PGj Z_{j-1}\| \\
                      &=&  \sum_{j=1}^l\|\PTp ({1\over q_j}\PGj Z_{j-1}-Z_{j-1})\| \\
                      &\leq & \sum_{j=1}^l\|{1\over q_j}\PGj Z_{j-1}-Z_{j-1}\| \\
                      &\leq & \sum_{j=1}^l C\sqrt{{{n \log n}\over q_j}}\|Z_{j-1}\|_\infty\\
                      &\leq & C\sqrt{{{n \log n}}}(\sum_{j=3}^l{1\over {2^{j-1} \log n{\sqrt{q_j}}}}+{1\over {2 \sqrt{\log n}{\sqrt{q_2}}}}+{1\over\sqrt{q_1}})\|Z_0\|_\infty\\
                     &\leq&C{\sqrt{n \mu r\log n }\over{n\sqrt{\rho}}}\leq{1\over {8\sqrt{\log n}}},
\end{eqnarray*}
provided $C_\rho$ is sufficiently large.\\
\\
Third, we have $\|\lambda\PTp\PP_{\Omega'}(W)\|\leq \lambda\|\PP_{\Omega'}(W)\|$. Notice that $W_{ij}$ is an independent Rademacher sequence independent of $\Omega'$. By Lemma \ref{lemma 4.3}, we have
\[
\|{{2s\rho}\over{1-\rho+2s\rho}}W-\PP_{\Omega'}(W)\|\leq C_0'\sqrt{np\log n}\|W\|_\infty
\]
with high probability provided ${{2s\rho}\over{1-\rho+2s\rho}}\leq p$ and $p\geq C_0 {{\log n}\over n}$. By Theorem 3.9 of \cite{Vershynin10}, we have $\|W\|_\infty\leq C_1\sqrt{n}$ with high probability. Therefore,
\[
\|\PP_{\Omega'}(W)\| \leq C_0'\sqrt{np\log n}+C_1\sqrt{n}{{2s\rho}\over{1-\rho+2s\rho}}.
\]
By choosing $p={\rho \over{C_2}}$ for some appropriate $C_2$, we have $\|\PP_{\Omega'}(W)\|\leq {\sqrt{n\rho\log n}\over 8}$, provided $C_\rho$ is large enough and $C_s$ is small enough.\\
\\
Fourth,
\begin{eqnarray*}
\|\PG Y\|_\infty&=&\|\PG \sum_j{1\over q_j}\PGj Z_{j-1}\|_\infty\\
                              &\leq& \sum_j{1\over q_j}\|Z_{j-1}\|_\infty\\
                              &\leq& (\sum_{j=3}^l {1\over q_j}{1\over {2^{j-1} \log n}}+{1\over q_2}{1\over {2 \sqrt{\log n}}}+{1\over q_1})\|Z_0\|_\infty\\
                              &\leq&C{\sqrt{\mu r}\over{n \rho}}\leq {{\lambda}\over{4\sqrt{\log n}}},
\end{eqnarray*}
provided $C_\rho$ is sufficiently large.
\end{proof}
Notice that in \cite{CLMW} the authors used a very similar golfing scheme. To compare these two methods, we use here a non-uniform sizes golfing scheme to achieve a result with fewer log factors. Moreover, unlike in \cite{CLMW} the authors used both golfing scheme and least square method to construct two parts of the dual matrix, here we only use golfing scheme. Actually the method to construct the dual matrix in \cite{CLMW} cannot be applied directly to our problem when $\rho=O(r\log^2 n/n)$.

\subsection*{Acknowledgements}
I am grateful to my Ph.~D.~advisor, Emmanuel Cand\`es, for his
encouragements and his help in preparing this manuscript.

\small 
\bibliographystyle{plain}
\bibliography{ref}

\begin{thebibliography}{10}

\bibitem{ANW11}
A.~Agarwal, S.~Negahban, and M.~Wainwright.
\newblock Noisy matrix decomposition via convex relaxation: Optimal rates in
  high dimensions.
\newblock {\em in Proc. 28th Inter. Conf. Mach. Learn. (ICML).}, pages
  1129--1136, 2011.

\bibitem{AW02}
R.~Ahlswede and A.Winter.
\newblock Strong converse for identification via quantum channels.
\newblock {\em IEEE Trans. Inform. Theory}, 48(3):569 -- 579, 2002.

\bibitem{BDDW07}
R.~Baraniuk, M.~Davenport, R.~DeVore, and M.~Wakin.
\newblock A simple proof of the restricted isometry property for random
  matrices.
\newblock {\em Constructive Approximation}, 28(3):253--263, 2008.

\bibitem{CLMW}
E.~Cand{\`e}s, X.~Li, Y.~Ma, and J.~Wright.
\newblock Robust principal component analysis?
\newblock {\em Journal of ACM}, 58(3), 2011.

\bibitem{CP09Proceedings}
E.~Cand\`es and Y.~Plan.
\newblock Matrix completion with noise.
\newblock {\em Proceedings of the {IEEE}}, 2009.

\bibitem{CP09}
E.~Cand\`es and Y.~Plan.
\newblock Near-ideal model selection by $\ell_1$ minimization.
\newblock {\em Ann. Statist.}, 37(5A):2145--2177, 2009.

\bibitem{CP10}
E.~Cand{\`e}s and Y.~Plan.
\newblock A probabilistic and ripless theory of compressed sensing.
\newblock {\em IEEE Transactions on Information Theory}, 57(11):7235--7254,
  2011.

\bibitem{CR09}
E.~Cand{\`{e}s} and B.~Recht.
\newblock Exact matrix completion via convex optimzation.
\newblock {\em Foundations of Computational Mathematics}, 9(6), 2009.

\bibitem{CRT06}
E.~Cand{\`e}s, J.~Romberg, and T.~Tao.
\newblock Robust uncertainty principles: exact signal reconstruction from
  highly incomplete frequency information.
\newblock {\em IEEE Trans. Inform. Theory}, 52(2):489--509, 2006.

\bibitem{CRT06CPAM}
E.~Cand\`{e}s, J.~Romberg, and T.~Tao.
\newblock Stable signal recovery from incomplete and inaccurate measurements.
\newblock {\em Communications on Pure and Applied Mathematics},
  59(8):1207--1223, 2006.

\bibitem{CT05}
E.~Cand{\`{e}s} and T.~Tao.
\newblock Decoding by linear programming.
\newblock {\em {IEEE} Trans. Information Theory}, 51(12), 2005.

\bibitem{CT10}
E.~Cand\`es and T.~Tao.
\newblock The power of convex relaxation: Near-optimal matrix completion.
\newblock {\em IEEE Trans. Inform. Theory}, 56(5):2053--2080, 2010.

\bibitem{CSPW09IFAC}
V.~Chandrasekaran, S.~Sanghavi, P.~Parrilo, and A.~Willsky.
\newblock Sparse and low-rank matrix decompositions.
\newblock {\em in 15th IFAC Sypmposium on System Identification (SYSID)}, 2009.

\bibitem{CSPW09}
V.~Chandrasekaran, S.~Sanghavi, P.~Parrilo, and A.~Willsky.
\newblock Rank-sparsity incoherence for matrix decomposition.
\newblock {\em SIAM J. on Optimization}, 21(2):572--596, 2011.

\bibitem{CDS98}
S.~Chen, D.~Donoho, and M.~Saunders.
\newblock Atomic decomposition by basis pursuit.
\newblock {\em SIAM J. Sci. Comput.}, 20(1):33--61, 1998.

\bibitem{CJSC11}
Y.~Chen, A.~Jalali, S.~Sanghavi, and C.~Caramanis.
\newblock Low-rank matrix recovery from errors and erasures.
\newblock {\em ISIT}, 2011.

\bibitem{DS01}
K.~Davidson and S.~Szarek.
\newblock Local operator theory, random matrices and banach spaces.
\newblock {\em Handbook of the Geometry of Banach Spaces}, I(8):317--366, 2001.

\bibitem{Donoho06CPAM}
D.Donoho.
\newblock For most large underdetermined systems of linear equations the
  minimal l1-norm solution is also the sparsest solution.
\newblock {\em Communications on Pure and Applied Mathematics}, 59(6):797--829,
  2006.

\bibitem{Donoho06}
D.~Donoho.
\newblock Compressed sensing.
\newblock {\em IEEE Trans. Inform. Theory}, 52(4):1289 -- 1306, 2006.

\bibitem{Fazel02}
M.~Fazel.
\newblock Matrix rank minimization with applications.
\newblock {\em Ph.D Thesis}, 2002.

\bibitem{Gross09}
D.~Gross.
\newblock Recovering low-rank matrices from few coefficients in any basis.
\newblock {\em IEEE Trans. on Information Theory}, 57(3):1548--1566, 2011.

\bibitem{GLFBE10}
D.~Gross, Y-K. Liu, S.Flammia, S.~Becker, and J.Eisert.
\newblock Quantum state tomography via compressed sensing.
\newblock {\em Physical Review Letters}, 105(15), 2010.

\bibitem{HBRN08}
J.~Haupt, W.~Bajwa, M.~Rabbat, and R.~Nowak.
\newblock Compressed sensing for networked data.
\newblock {\em Signal Processing Magazine, IEEE}, 25(2):92 -- 101, 2008.

\bibitem{HKZ11}
D.~Hsu, S.~Kakade, and T.~Zhang.
\newblock Robust matrix decomposition with sparse corruptions.
\newblock {\em Information Theory, IEEE Transactions on}, 57(11):7221--7234,
  2011.

\bibitem{Tropp11}
J.Tropp.
\newblock User-friendly tail bounds for sums of random matrices.
\newblock {\em Found. Comput. Math.}, 2011.

\bibitem{KMO10}
R.~Keshavan, A.~Montanari, and S.~Oh.
\newblock Matrix completion from a few entries.
\newblock {\em IEEE Trans. Inform. Theory}, 56(6):2980--2998, 2010.

\bibitem{LBDB11}
J.~Laska, P.~Boufounos, M.~Davenport, and R.~Baraniuk.
\newblock Democracy in action: Quantization, saturation, and compressive
  sensing.
\newblock {\em Applied and Computational Harmonic Analysis}, 31(3):429--443,
  2011.

\bibitem{LDB09}
J.~Laska, M.~Davenport, and R.~Baraniuk.
\newblock Exact signal recovery from sparsely corrupted measurements through
  the pursuit of justice.
\newblock {\em Asilomar Conference on Signals Systems and Computers}, 2009.

\bibitem{LWW10}
Z.~Li, F.~Wu, and J.~Wright.
\newblock On the systematic measurement matrix for compressed sensing in the
  presence of gross errors.
\newblock {\em Data Compression Conference}, pages 356--365, 2010.

\bibitem{NT11}
N.~Nguyen and T.~Tran.
\newblock Exact recoverability from dense corrupted observations via
  l$_{\mbox{1}}$ minimization.
\newblock {\em preprint}, 2011.

\bibitem{NNT11}
N.~Ngyuen, N.~Nasrabadi, and T.~Tran.
\newblock Robust lasso with missing and grossly corrupted observations.
\newblock {\em preprint}, 2011.

\bibitem{Recht09}
B.~Recht.
\newblock A simpler approach to matrix completion.
\newblock {\em Journal of Machine Learning Research}, 12:3413--3430, 2011.

\bibitem{RFP10}
B.~Recht, M.~Fazel, and P.~Parillo.
\newblock Guaranteed minimum-rank solutions of linear matrix equations via
  nuclear norm minimization.
\newblock {\em SIAM Review}, 52(3), 2010.

\bibitem{Romberg09}
J.~Romberg.
\newblock Compressive sensing by random convolution.
\newblock {\em SIAM J. Imaging Sciences}, 2(4):1098--1128, 2009.

\bibitem{Tib96}
R.Tibshirani.
\newblock Regression shrinkage and selection via the lasso.
\newblock {\em J. Royal Statist. Soc. B.}, 58(1):267--288, 1996.

\bibitem{Rudelson99}
M.~Rudelson.
\newblock Random vectors in the isotropic position.
\newblock {\em J. of Functional Analysis}, 164(1):60--72, 1999.

\bibitem{RV06}
M.~Rudelson and R.~Vershynin.
\newblock On sparse reconstruction from fourier and gaussian measurements.
\newblock {\em Communications on Pure and Applied Mathematics},
  61(8):1025--1045, 2008.

\bibitem{SKPB11}
C.~Studer, P.~Kuppinger, G.~Pope, and H.~B{\"o}lcskei.
\newblock Recovery of sparsely corrupted signals.
\newblock {\em preprint}, 2011.

\bibitem{Vershynin10}
R.~Vershynin.
\newblock Introduction to the non-asymptotic analysis of random matrices.
\newblock {\em Chapter 5 of the book Compressed Sensing, Theory and
  Applications, ed. Y. Eldar and G. Kutyniok. Cambridge University Press},
  pages 210--268, 2012.

\bibitem{WM10}
J.~Wright and Y.~Ma.
\newblock Dense error correction via $\ell_1$-minimization.
\newblock {\em IEEE Transactions on Information Theory}, 56(7):3540 -- 3560,
  2010.

\bibitem{WYGSM09}
J.~Wright, A.~Y. Yang, A.~Ganesh, S.~Sastry, and Y.~Ma.
\newblock Robust face recognition via sparse representation.
\newblock {\em IEEE Trans. Pattern Anal. Mach. Intell.}, 31(2):210Ð227, 2009.

\bibitem{WGSMWM10}
L.~Wu, A.~Ganesh, B.~Shi, Y.~Matsushita, Y.~Wang, and Y.~Ma.
\newblock Robust photometric stereo via low-rank matrix completion and
  recovery.
\newblock {\em Proceedings of the 10th Asian conference on Computer vision},
  Part III, 2010.

\bibitem{XCS10}
H.~Xu, C.~Caramanis, and S.~Sanghavi.
\newblock Robust pca via outlier pursuit.
\newblock {\em in Ad. Neural Infor. Proc. Sys. (NIPS)}, pages 2496--2504, 2010.

\end{thebibliography}

\end{document}